\documentclass[journal]{IEEEtran}
\usepackage{graphicx}
\usepackage{subcaption}
\usepackage{float}
\usepackage{physics}
\usepackage{cite}
\usepackage{amsmath,amsthm,amsfonts,amssymb,amscd}
\usepackage{amsmath}
\usepackage{xcolor}
\usepackage{bm}
\usepackage{cuted}
\usepackage{algorithm}
\usepackage[noend]{algpseudocode}
\usepackage{multirow}
\usepackage{verbatim}
\usepackage[hidelinks]{hyperref}

\usepackage{mathtools}
\newtheorem{theorem}{Theorem}

\begin{document}
\title{System-wide Instrument Transformer Calibration and Line Parameter Estimation Using PMU Data}
%
%
%

\author{Antos~Cheeramban~Varghese,~\IEEEmembership{Student Member,~IEEE,}
        and Anamitra~Pal,~\IEEEmembership{Senior~Member,~IEEE}
\thanks{This work was supported in part by the National Science Foundation
(NSF) Grants OAC 1934766 and ECCS 2145063.\\
The authors are associated with the School of Electrical, Computer, and
Energy Engineering, Arizona State University, Tempe, AZ 85287,
USA.}
}

%
%

\markboth{Under Peer Review}%
{Shell \MakeLowercase{\textit{et al.}}: Bare Demo of IEEEtran.cls for IEEE Journals}
%



\maketitle

\begin{abstract}
Uncalibrated instrument transformers (ITs) can degrade the performance of downstream applications that rely on the voltage and current measurements that ITs provide.
It is also well-known
that phasor measurement unit (PMU)-based system-wide IT calibration and line parameter estimation (LPE) are interdependent problems.
In this paper, we present a statistical framework for solving the \underline{s}imultaneous \underline{L}PE and \underline{I}T \underline{c}alibration (SLIC) problem using synchrophasor data.
The proposed approach not only avoids the need for a ``perfect'' IT by judiciously placing a revenue quality meter (which is an expensive but non-perfect IT), but also accounts for the variations typically occurring in the line parameters. 
The results obtained using the IEEE 118-bus system
as well as actual power system data demonstrate the high accuracy, robustness, and practical utility of the proposed approach. 
\end{abstract}


\IEEEpeerreviewmaketitle

\section{Introduction}
\label{Intro}

\IEEEPARstart{I}{nstrument} transformers (ITs), which comprise voltage transformers (VTs) and current transformers (CTs), are key components of the power system measurement infrastructure.
They are responsible for \emph{accurately} scaling down the very high voltages and currents to levels that can be safely handled by the power system's instrumentation system. %
If the ITs are not able to faithfully reproduce the scaled down version of the primary-side signal on their secondary side, it will impact critical power system applications, such as relay operation, state estimation, and fault location \cite{phadke2017synchronized}.
The degradation in IT performance can occur due to environmental conditions as well as aging.
In this paper, degradation due to the non-ideal scaling of the ITs is expressed 
in terms of
\textit{ratio error} (RE)\footnote{Note that we are using the term RE to denote both the magnitude as well as the angle component of the non-ideal scaling; i.e., RE is a complex number. We use $\alpha$ to denote RE in this paper, and hence $\alpha \in \mathbb{C}$.}.
The range in which REs lie is given
in IEEE C57.13 Standard \cite{IEEE_C57_13_2016_std_for_ITs}.
IT calibration is
the process of off-setting the impact of REs on downstream applications by finding suitable \textit{correction factors}\footnote{Correction factors (denoted by $\tau$) are the inverse of the REs; i.e., $\tau = 1/\alpha$.}.
The traditional methods of IT calibration require external hardware (e.g., transformer comparator, highly accurate calibrator) \cite{brandolini2009simple,crotti2017industrial,siegenthaler2017computer}.
Consequently, such \emph{hard calibration} methods are difficult to apply when the substations are in-service.
With the introduction of phasor measurement units (PMUs), \emph{soft calibration} methods were proposed \cite{shi2012adaptive,zhou2012calibrating,chatterjee2018error}, which calibrated the ITs remotely using synchrophasor data and could be done at any time. 
However, a common drawback of these soft calibration methods was the need for prior
knowledge of the line parameters.
Now, line parameters can be estimated from PMU data \cite{varghese2022transmission}.
However,
if PMU measurements from \textit{uncalibrated} ITs are used for line parameter estimation (LPE), then the estimation quality will be poor \cite{mishra2015kalman,mansani2018estimation}. Thus, LPE and IT calibration constitutes a ``chicken-and-egg" problem. 

The first attempt at solving this ``chicken-and-egg"
problem was made in \cite{wu2015simultaneous}. However, it assumed the from-end ITs were ideal. It also reported a very high sensitivity to measurement noise and flagged it for future research.
Next, a \textit{bias error detection} (BED) test was proposed \cite{khandeparkar2016detection} to estimate line parameters and detect bias in ITs using voltage and current measurements from PMUs.
The BED test was improved upon
in \cite{goklani2020instrument} through the \textit{midpoint voltage} (MPV) test, which estimated line parameters and determined if the ITs were bad or not. However, both the BED test and the MPV test required the from-end ITs to be error-free.

Based on our literature survey, we have identified four practical aspects of the system-wide \underline{s}imultaneous \underline{L}PE and \underline{I}T \underline{c}alibration (SLIC) problem, as described below:
\begin{itemize}
    \item \textit{Reliance on a ``perfect" IT:} 
    Many prior works assumed the presence of an ideal/error-free IT, which is not realistic.
    The bulk power system (BPS) does have revenue quality meters (RQMs) \cite{pal2015online}, such as electro-optic VTs (EOVTs) and magneto-optic CTs (MOCTs), whose REs are considerably smaller than those of regular ITs\footnote{For example, a 0.5 accuracy class IT has a maximum magnitude error of 0.5\%, and a maximum angle error of 0.9 crad for CTs and 0.6 crad for VTs \cite{sitzia2022enhanced,pegoraro2022pmu}. Conversely, the maximum magnitude error of a RQM is 0.15\% and maximum angle error is 0.2 crad \cite{IEEE_C57_13_2016_std_for_ITs}.}.
    However, as RQMs are very expensive, it is best if system-wide SLIC is done with very few RQMs (ideally, one).
    \item \textit{Line parameter variations:} 
    Although line parameters change over time, 
    they lie within $\pm 30\%$ of their database values \cite{kusic2004measurement}.
    This knowledge should be exploited for solving the SLIC problem. 
    \item \textit{PMU data usage:} 
    The positive-sequence measurements produced by PMUs are most commonly used by power utilities for decision-making purposes \cite{wang2019transmission}. Therefore, it is valuable to determine
    the equivalent positive-sequence correction factors to compensate for the equivalent positive-sequence REs.
    \item \textit{PMU device errors:} The PMU device also adds an error (to the PMU output) that is independent of the non-ideal scaling of the ITs. This error 
    manifests as an additive Gaussian noise in the positive-sequence in Cartesian coordinates (see Fig. 3 of \cite{wang2019transmission}).
\end{itemize}

Considering these four practical aspects of the SLIC problem, this paper proposes a statistical framework to solve the problem.
Our main contributions are as follows:
\begin{enumerate}
    \item We develop a \textit{quantization procedure} 
    to uniquely identify the line parameters of a branch.
    To the best of our knowledge, this is the first paper that presents such a procedure to solve
    the SLIC problem.
    \item We formulate an algorithm that leverages the higher accuracy of RQMs for system-wide SLIC.
    This algorithm accurately estimates correction factors of all the ITs in a \textit{connected tree}\footnote{A connected tree is a set of branches defined as `From bus'-`To bus', such that every bus can be accessed from every other bus via the branches of the tree, and each bus has PMU. The connected tree is denoted by $\mathcal{L}$ in the paper.} 
    with just
    one RQM.
    \item We devise an algorithm to find \textit{that} RQM location (in case the connected tree does not have one already) which gives the least error for solving the SLIC problem. 
\end{enumerate}

Finally, we demonstrate: (a) \textit{high accuracy and robustness} by considering different levels of PMU noise and classes of ITs in the IEEE 118-bus system, 
and (b) \textit{practical utility} by 
solving the SLIC problem 
for a power company located in the U.S. Eastern Interconnection using field PMU data.


\section{Modeling the SLIC Problem for One Branch}
\label{SLIC Modeling}




\subsection{Basic Formulation}
\label{BasicForm}
Fig. \ref{Medium_length_line_Segment_pi_with_IT} shows a $\pi$-model of a medium-length transmission line that has CTs and VTs at both ends.
Here, $V$ and $I$ represent the complex voltage and current measurements from the PMUs, with $p$ and $q$ indicating the `from' and `to' ends. The line parameters to be estimated are resistance, \textcolor{black}{$r_{pq} \in \mathbb{R}$}, reactance, \textcolor{black}{$x_{pq} \in \mathbb{R}$}, and shunt susceptance present at each end, \textcolor{black}{$b_{pq} \in \mathbb{C}$}. These parameters for the branch $p$-$q$ are collectively represented by $\eta_{pq} = [r_{pq}, x_{pq}, b_{pq}]$. Next, we specify two derived variables: series impedance,  \textcolor{black}{$z_{pq} \in \mathbb{C}$} ($z_{pq}  = r_{pq} + j x_{pq}$), and series admittance,  \textcolor{black}{$y_{pq} \in \mathbb{C}$} ($y_{pq}  = {1}/{z_{pq}} $).
To solve the SLIC problem, we must
estimate the line parameters 
and the IT correction factors  
of every branch in the connected tree.
The relation between the line parameters and PMU measurements from either end of the branch $p$-$q$ can be obtained using Kirchhoff's laws as shown below:
\begin{equation}
\label{eqn:TLPE_Basic_KL_SLIC}
         \begin{aligned}
         I_{pq}^* &= b_{pq} V_{pq}^* + (V_{pq}^* - V_{qp}^*) y_{pq}\\
         I_{qp}^* &= b_{pq} V_{qp}^* - (V_{pq}^* - V_{qp}^*) y_{pq}
         \end{aligned}
\end{equation}
where the superscript `$*$' denotes the true value.
Eq. \eqref{eqn:TLPE_Basic_KL_SLIC} is valid
when (a) ITs have RE of unity (i.e., $\alpha = 1 + j0$), and (b) PMU measurements are noise-free.
Since these conditions are not met in practice, we modify \eqref{eqn:TLPE_Basic_KL_SLIC} to account for the non-idealities encountered in actual power system operation.

\begin{figure}[H]
    \centering
    \includegraphics[width=0.485\textwidth]{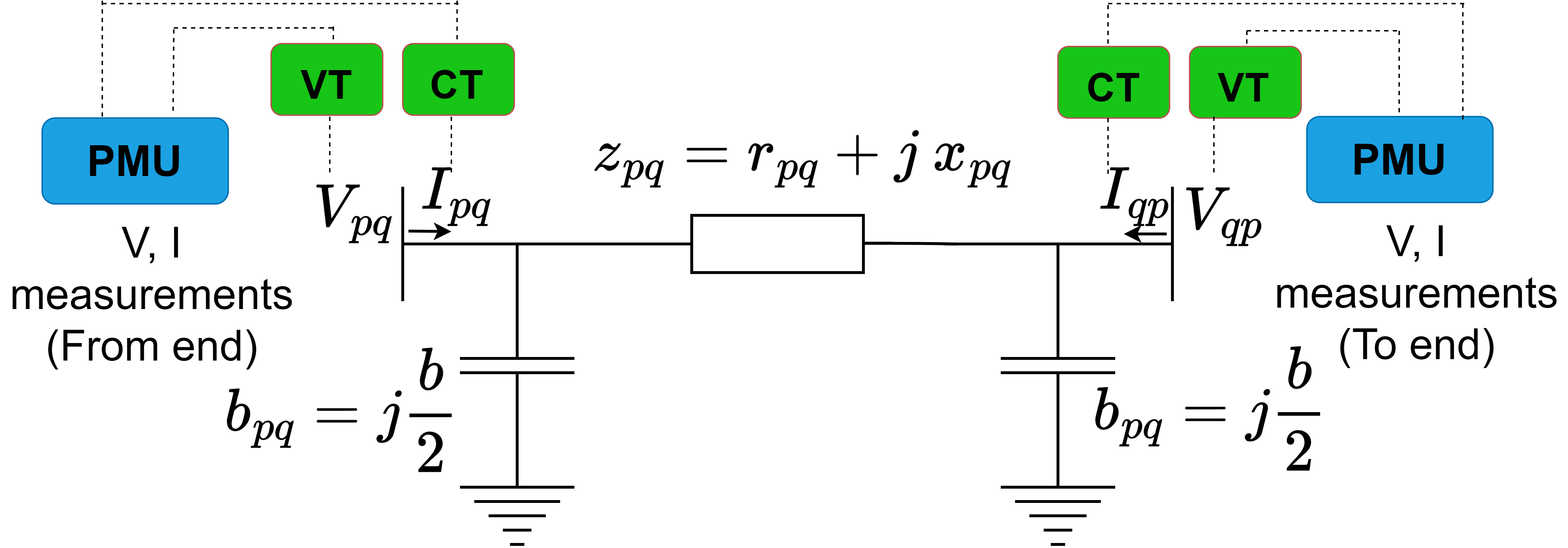}
    \caption{\textcolor{black}{$\pi$-model of transmission line used for SLIC}}\label{Medium_length_line_Segment_pi_with_IT}
\end{figure}

When the ITs have non-unity REs, they will appear as an \textit{unknown} multiplication factor with the 
true phasor value
\cite{phadke2017synchronized}.
Additionally, for the time period over which the measurements are collected for SLIC, the REs will be constant 
quantities. This is because the REs change at a much slower pace than the speed at which PMUs produce phasor measurements \cite{zhou2012calibrating}.
In light of these facts, we derive the following relation between the true values and the noisy measurements:
\begin{equation}
\label{eqn:Relation after IT - true}
         \begin{aligned}
         V_{pq} &= \alpha_{V_{pq}} V_{pq}^* +V_{pq_e} \\
         V_{qp} &= \alpha_{V_{qp}} V_{qp}^* +V_{qp_e} \\
         I_{pq} &= \alpha_{I_{pq}} I_{pq}^* +I_{pq_e} \\
         I_{qp} &= \alpha_{I_{qp}} I_{qp}^* + I_{qp_e}.
         \end{aligned}
\end{equation}

In \eqref{eqn:Relation after IT - true}, $V_{pq_e}$, $V_{qp_e}$, $I_{pq_e}$, and $I_{qp_e}$ denote the noise
added by the PMU device. 
It can be observed from \eqref{eqn:Relation after IT - true} that the PMU measurement has a \textit{composite} noise model, where the additive component is a random variable and the multiplicative component is an unknown constant.
It is also clear from \eqref{eqn:Relation after IT - true} that using only the measurement information, one cannot determine the REs because the ``true value" is also unknown.
Now, since the correction factor, $\tau$, is the inverse of the RE, $\alpha$, we can write the following equation:
\begin{equation}
\label{eqn:Relation after IT - true_Ks}
         \begin{aligned}
         V_{pq}^* &= \tau_{V_{pq}} (V_{pq} -V_{pq_e})  \\
         V_{qp}^* &= \tau_{V_{qp}} (V_{qp} -V_{qp_e}) \\
         I_{pq}^* &= \tau_{I_{pq}} (I_{pq} -I_{pq_e}) \\
         I_{qp}^* &= \tau_{I_{qp}} (I_{qp} - I_{qp_e}).
         \end{aligned}
\end{equation}


Next, by substituting \eqref{eqn:Relation after IT - true_Ks} in \eqref{eqn:TLPE_Basic_KL_SLIC} and rearranging, we get
\begin{equation}
\label{eqn:SLIC_T2_homogeneous}
         \begin{aligned}
          &\tau_{I_{pq}} (I_{pq}-I_{pq_e}) -  b_{pq}  \tau_{V_{pq}} (V_{pq} -V_{pq_e}) \\
          &-  y_{pq} \tau_{V_{pq}} (V_{pq} -V_{pq_e}) + y_{pq} \tau_{V_{qp}} (V_{qp} - V_{qp_e}) = 0\\
          &\tau_{I_{qp}} (I_{qp} - I_{qp_e}) -  b_{pq}  \tau_{V_{qp}} ( V_{qp} - V_{qp_e}) \\
          &+ y_{pq} \tau_{V_{pq}} (V_{pq}  - V_{pq_e}) - y_{pq} \tau_{V_{qp}}  (V_{qp} - V_{qp_e}) = 0.
         \end{aligned}
\end{equation}

The drawback of \eqref{eqn:SLIC_T2_homogeneous} is that it is a \textit{homogenous system of equations} due to which a unique solution cannot be obtained. 
The homogeneity occurs because every term in \eqref{eqn:SLIC_T2_homogeneous} has at least one of the parameters to be estimated (line parameters or IT correction factors) as a multiplier.
The next sub-section describes how an RQM can help circumvent this issue.

\subsection{Overcoming Homogeneity Issue with RQM}
One way to achieve a non-homogeneous linear system of equations is by rendering at least one term in \eqref{eqn:SLIC_T2_homogeneous} independent of the correction factor by which it is multiplied. This can be done by dividing each term in \eqref{eqn:SLIC_T2_homogeneous} by a common correction factor, to produce \textit{correction factor ratios}. 
Then, the term that was originally multiplied by this common divisor becomes parameter-independent and can be relocated to the right-hand side (RHS) of \eqref{eqn:SLIC_T2_homogeneous}, yielding a non-homogeneous linear regression model. However, this has the drawback that the parameters estimated from this set of equations will be the correction factor ratios instead of the correction factors of the ITs. 
This drawback was also identified in prior research on PMU-based IT calibration
\cite{zhou2012calibrating,wu2015simultaneous}.
However, they resolved it by assuming a perfect IT in the system, which is an unrealistic assumption. 
The proposed approach overcomes this drawback by using an RQM.


RQMs are very high quality but very expensive ITs that are often placed at strategic locations in the system (such as large generator buses, ends of tie-lines).
They have at least $3\mathrm{x}$ better accuracies than the conventional ITs, resulting in the REs of the RQMs lying within a very narrow range centered around $1+j0$ \cite{IEEE_C57_13_2016_std_for_ITs}.
If we now assume that the VT at the $p$-end of the branch $p$-$q$ has been upgraded to an EOVT (which is one type of RQM), then we can divide each term in \eqref{eqn:SLIC_T2_homogeneous} with $\tau_{V_{pq}}$ to obtain the following equation:
\begin{equation}
\label{eqn:SLIC_T3_non_homogeneous}
         \begin{aligned}
          & (y_{pq} + b_{pq})   (V_{pq} -V_{pq_e}) - y_{pq} \frac{\tau_{V_{qp}}}{\tau_{V_{pq}}} (V_{qp} - V_{qp_e})  \\ 
          & - \frac{\tau_{I_{pq}}}{\tau_{V_{pq}}} (I_{pq}-I_{pq_e}) = 0 \\
          & (y_{pq} + b_{pq})  \frac{\tau_{V_{qp}}}{\tau_{V_{pq}}} ( V_{qp} - V_{qp_e}) - \frac{\tau_{I_{qp}}}{\tau_{V_{pq}}} (I_{qp} - I_{qp_e}) \\
          &  = y_{pq} (V_{pq}  - V_{pq_e}) .
         \end{aligned}
\end{equation}

Next, we pre-multiply \eqref{eqn:SLIC_T3_non_homogeneous} by $z_{pq} (= 1/y_{pq})$ to make the RHS of the second sub-equation parameter-independent. Finally, we pre-multiply the resulting first sub-equation of \eqref{eqn:SLIC_T3_non_homogeneous} by $W_{pq}$, where $W_{pq} = (1 + z_{pq} b_{pq})$, to get the following equation:

\begin{equation}
\label{eqn:SLIC_final_form_single_time_instant}
         \begin{aligned}
          &{W_{pq}}^2   (V_{pq} - V_{pq_e})  -  W_{pq}  \frac{\tau_{V_{qp}}}{\tau_{V_{pq}}} (V_{qp}- V_{qp_e}) \\
          &-  W_{pq} z_{pq} \frac{\tau_{I_{pq}}}{\tau_{V_{pq}}} (I_{pq} - I_{pq_e} )     =  0 \\
          & W_{pq}  \frac{\tau_{V_{qp}}}{\tau_{V_{pq}}}  (V_{qp}- V_{qp_e})\\
           &- z_{pq}  \frac{\tau_{I_{qp}}}{\tau_{V_{pq}}}  (I_{qp} - I_{qp_e})  =   (V_{pq} - V_{pq_e}) .
         \end{aligned}
\end{equation}

Eq. \eqref{eqn:SLIC_final_form_single_time_instant} describes a relation between the measurements and the parameters to be estimated for a single time-instant.
Now, by concatenating measurements from $n$ different time-instants (i.e., different operating conditions), an over-determined system of equations can be obtained as shown below:
\begin{equation}
\begin{aligned}
\label{eqn:SLIC_real_abstraction}
\begin{bmatrix}  (D+D_e)\end{bmatrix} \: \theta =\begin{bmatrix}  (c+c_e) \end{bmatrix}
 \end{aligned}
\end{equation}
where
\begin{equation}
\begin{aligned}
\label{eqn:SLIC_complex_D_for_n_samples}
 D = \begin{bmatrix}  V_{pq} (1) &  - V_{qp} (1)    &  - I_{pq} (1)  & 0 \\  0 &   V_{qp} (1)  &   0 &   -I_{qp}(1) \\   V_{pq} (2) &  - V_{qp} (2)    &  - I_{pq} (2)  & 0 \\  0 &   V_{qp} (2)  &   0 &   -I_{qp}(2) \\ \vdots &   \vdots    &   \vdots  &  \vdots \\ V_{pq} (n) &  - V_{qp} (n)    &  - I_{pq} (n)  & 0 \\  0 &   V_{qp} (n)  &   0 &   -I_{qp}(n) \\ \end{bmatrix} \\ 
 \end{aligned}
\end{equation}
\begin{equation}
\begin{aligned}
\label{eqn:SLIC_complex_c_for_n_samples}
c = \begin{bmatrix}  0   &    V_{pq} (1) &   0   &    V_{pq} (2) & \dots & 0   &    V_{pq} (n) \end{bmatrix}^\intercal 
 \end{aligned}
\end{equation}
\begin{equation}
\label{SLIC:theta_parameters_Expression}
    \begin{aligned}
     \theta = \begin{bmatrix}  {W_{pq}}^2   &    W_{pq}  \frac{\tau_{V_{qp}}}{\tau_{V_{pq}}} &  W_{pq} z_{pq}  \frac{\tau_{I_{pq}}}{\tau_{V_{pq}}} &  z_{pq} \frac{\tau_{I_{qp}}}{\tau_{V_{pq}}} \end{bmatrix}^\intercal
    \end{aligned}
\end{equation}
and $D_e$ and $c_e$ are the noise in the corresponding elements of $D$ and $c$, respectively.
Note that the elements in \eqref{eqn:SLIC_real_abstraction}-\eqref{SLIC:theta_parameters_Expression} are in the complex domain. When expressed in Cartesian coordinates, $D_e$ and $c_e$ (which correspond to noises added by PMUs) will follow a zero-mean Gaussian distribution \cite{wang2019transmission}.
Moreover, the presence of noise in both the dependent variable ($c$) and the independent variable ($D$) categorizes \eqref{eqn:SLIC_real_abstraction} as an \textit{errors-in-variables} (EIV) estimation problem. 
The solution to a linear EIV problem with Gaussian noise can be found using a statistical method called \textit{total least squares} (TLS) \cite{markovsky2007overview}.

Despite the solution of \eqref{eqn:SLIC_real_abstraction} yielding \textit{eight} estimates in the real domain, we cannot use it to uniquely identify the line parameters and correction factor ratios of branch $p$-$q$. This is explained as follows: 
Let the IT correction factor ratios be denoted by $\Gamma_{pq}$, i.e., $\Gamma_{pq} = [\frac{\tau_{V_{qp}}}{\tau_{V_{pq}}} , \frac{ \tau_{I_{pq}}}{\tau_{V_{pq}}} , \frac{ \tau_{I_{qp}}}{\tau_{V_{pq}}} ]$. Then, by definition, $\Gamma_{pq} \in \mathbb{C}^{3 \times 1}$.
Now, since $\eta_{pq} \in \mathbb{R}^{3 \times 1}$ (see first paragraph of Section \ref{BasicForm}), we have \textit{nine} 
unknown parameters to be estimated in the real domain (i.e., the problem is under-determined). In the next sub-section, we propose a quantization procedure to address the problem of uniquely identifying the line parameters and correction factor ratios from \eqref{eqn:SLIC_real_abstraction}-\eqref{SLIC:theta_parameters_Expression}. 
As the focus in this section is on a single branch, we refer to the problem as the individual branch (IB)-SLIC problem.

\subsection{Quantization Procedure for Solving the IB-SLIC Problem}
\label{QP4IBSLIC}

The proposed quantization procedure makes use
of the fact that the first element of $\theta$ in \eqref{SLIC:theta_parameters_Expression}, namely, $\theta_1 = {W_{pq}}^2$, is in-dependent of $\Gamma$.
Hence, if $\hat{\eta}_{pq}$ can be obtained from $\hat{\theta}_1$, then the correction factor ratios can be estimated from the other terms of $\hat{\theta}$.
To obtain $\hat{\eta}_{pq}$ (from $\hat{\theta}_1$), we exploit how line parameters vary in an actual power system, as explained below.

Particularly for the BPS, which is the focus of this paper, power utilities maintain a database that contains historical values of the line parameters. Although the parameters change over time, they lie within $\pm30\%$ of their database values \cite{kusic2004measurement}, implying that the variations are bounded.
Furthermore, the variations are also similar \cite{cecchi2011incorporating}.
In such a scenario, we `bin' the variations as a function of the database value in the following way: Define 
$r_{pq}(m) = r_{\dagger} (1 + \delta_{r} m)$, $x_{pq}(m) = x_{\dagger} (1 + \delta_{x} m)$,  $b_{pq}(m) = b_{\dagger} (1 + \delta_{b} m)$, and $\eta_{\dagger}=[r_{\dagger},x_{\dagger},b_{\dagger}]$, where the symbol $\dagger$ in the subscript indicates database value of the corresponding line parameter of branch $p$-$q$, $\delta$ denotes the normalized quantization step-size for the corresponding parameter of branch $p$-$q$, and $m$ is the bin number. The bounds are enforced by ensuring $\delta m \in [-0.3,0.3]$.
Once the bins are created, we define $W_{pq}(m)$ as a function of $r_{pq}(m)$, $x_{pq}(m)$, and $b_{pq}(m)$, as shown below:
\begin{equation}
\begin{aligned}
\label{f_W_definition}
W_{pq}(m) &= f_W(r_{pq}(m), x_{pq}(m), b_{pq}(m)) \\ 
& = 1+(r_{pq}(m)+jx_{pq}(m)) b_{pq}(m).
\end{aligned}
\end{equation}

The estimate of $W_{pq}$, $\widehat{W}_{pq}$, is already known from $\hat{\theta}_1$ by solving \eqref{eqn:SLIC_real_abstraction} using TLS. The optimal bin index ($m^*$), and thereby, $\hat{\eta}_{pq}$, are now found by determining the $W_{pq}(m)$ that is closest to $\widehat{W}_{pq}$ in the least-squares sense.
This quantization procedure is mathematically explained in Algorithm \ref{alg:SLIC_Algo_0}.

The proposed quantization procedure ensures tractability of the IB-SLIC problem, while being subject to the accuracy of the TLS solution. 
The procedure does introduce a trade-off between accuracy and computational burden; however, this can be controlled by modulating the quantization step-size.
A more crucial aspect is the guarantee of the unique identification of the line parameters; we prove this through Theorem \ref{Theorem_1} below.

\begin{algorithm}[H]
\caption{Quantization Procedure (QP)} \label{alg:SLIC_Algo_0}
\begin{algorithmic}[1] 
\State \textbf{Input: }  $D, c, \eta_{\dagger}, p, q$
\State \textbf{Output: }  $\hat{\eta}_{pq}$
\Procedure{QP }{$D, c, \eta_{\dagger}, p, q$}
    \State $\hat{\theta} = \text{TLS}(D, c)$
    \State $ \widehat{W}_{pq} = \sqrt{\hat{\theta}_1}$
    \State Create bins using $\eta_{\dagger}$:
    \State \quad $ r_{pq}(m) = r_{\dagger} (1+ \delta_{r} m)$ 
    \State \quad $ x_{pq}(m) = x_{\dagger} (1+ \delta_{x} m)$ 
    \State \quad $ b_{pq}(m) = b_{\dagger} (1+ \delta_{b} m)$, and 
    \State \quad $W_{pq}(m) =  f_W (r_{pq}(m), x_{pq}(m), b_{pq}(m))$
    \State Optimal bin index \mbox{$ m^* = \arg\underset{m }\min \| \widehat{W}_{pq} - W_{pq}(m)\|_2^2$}
    \State Assign the line parameter values based on the optimal bin index as
    \State \quad $\hat{r}_{pq} = r_{pq}(m^*)$
    \State \quad $\hat{x}_{pq} = x_{pq}(m^*)$
    \State \quad $\hat{b}_{pq} = b_{pq}(m^*)$
    \State \quad $\hat{\eta}_{pq} = [\hat{r}_{pq}, \hat{x}_{pq}, \hat{b}_{pq}]$
    \State \textbf{return} $\hat{\eta}_{pq}$
\EndProcedure
\end{algorithmic}
\end{algorithm}


\begin{theorem}
If $m_1$ and $m_2$ denote two bin numbers, then
$f_W(m_1) = f_W(m_2) \iff m_1 = m_2$.
\label{Theorem_1}
\end{theorem}

\begin{proof}
Please see Appendix \ref{appendix1} for the proof.
\end{proof}

Once the line parameter estimates for the line $p$-$q$ (namely, $\hat{\eta}_{pq}$) are found using the quantization procedure, they can be used to estimate the corresponding IT correction factor ratios (namely, $\hat{\Gamma}_{pq}$), as shown below,
\begin{equation}
\label{CFR-back calculation}
\begin{aligned}
\frac{\hat{\tau}_{V_{qp}}}{\tau_{V_{pq}}} =  \frac{ \hat{\theta}_2}{\hat{W}_{pq}}, \:
\frac{\hat{\tau}_{I_{pq}}}{\tau_{V_{pq}}} =  \frac{ \hat{\theta}_3}{ \hat{W}_{pq} \hat{z}_{pq}}, \:
\frac{\hat{\tau}_{I_{qp}}}{\tau_{V_{pq}}} =\frac{ \hat{\theta}_4 }{ \hat{z}_{pq}  }.
    \end{aligned}
\end{equation}


The overall procedure for solving the IB-SLIC
problem is now summarized in Algorithm \ref{alg:SLIC_Algo_1}.
The algorithm takes the `from' and `to' end bus information ($p$ and $q$), the noisy voltage and current phasor measurements ($V_{pq}, V_{qp}, I_{pq}, I_{qp}$), and the database values of the line parameters ($\eta_{\dagger}$) as inputs.
First, $D$ and $c$ are obtained for the $p$-$q$ branch using \eqref{eqn:SLIC_complex_D_for_n_samples} and \eqref{eqn:SLIC_complex_c_for_n_samples}. Then, using TLS, the optimal parameter estimate, $\hat{\theta}$, is found. Finally, for obtaining $\hat{\eta}_{pq}$ and $\hat{\Gamma}_{pq}$ from $\hat{\theta}$, the quantization procedure (Algorithm \ref{alg:SLIC_Algo_0}) and \eqref{CFR-back calculation} are utilized.

\begin{algorithm}[H]
\caption{Individual Branch (IB)-SLIC} \label{alg:SLIC_Algo_1}
\begin{algorithmic}[1] 
\State \textbf{Input: }  $p, q, V_{pq}, V_{qp}, I_{pq}, I_{qp}, \eta_{\dagger}$
\State \textbf{Output: }  $\hat{\eta}_{pq}$, $\hat{\Gamma}_{pq}$
\Procedure{IB-SLIC }{$p, q, V_{pq}, V_{qp}, I_{pq}, I_{qp}, \eta_{\dagger}$}
    \State Create $D$ using \eqref{eqn:SLIC_complex_D_for_n_samples}
    \State Create $c$ using \eqref{eqn:SLIC_complex_c_for_n_samples}
    \State Calculate $\hat{\eta}_{pq}$ using Algorithm \ref{alg:SLIC_Algo_0} 
    \State Calculate $\hat{\Gamma}_{pq}$ using \eqref{CFR-back calculation} and $\hat{\eta}_{pq}$ 
     \State \textbf{return}  $\hat{\eta}_{pq}$, $\hat{\Gamma}_{pq}$
\EndProcedure
\end{algorithmic}
\end{algorithm}

Algorithm \ref{alg:SLIC_Algo_1} estimates line parameters ($\eta_{pq}$) and correction factor ratios ($\Gamma_{pq}$) for branch $p$-$q$ in the network.
However, the eventual goal is to estimate the IT correction factors (henceforth, denoted by $\mathcal{T}_{pq}$, where $\mathcal{T}_{pq} = [\tau_{V_{pq}}, \tau_{V_{qp}} \tau_{I_{pq}}, \tau_{I_{qp}}]$).
Now, if the $p$-end of this branch has an EOVT, then 
we can solve \eqref{eqn:SLIC_real_abstraction} a large number of times (say, $M$), and take the average to calculate $\hat{\tau}_{V_{qp}}$, $\hat{\tau}_{I_{pq}}$, and $\hat{\tau}_{I_{qp}}$ as shown below:
%
\begin{equation}
\label{eqn:kappa_from_rho}
         \begin{aligned}
    \hat{\tau}_{V_{qp}}   &= \frac{1}{M} \sum_{j=1}^M \left( \frac{\hat{\tau}_{V_{qp}}}{\tau_{V_{pq}}}\right)_j , \:
    \hat{\tau}_{I_{pq}}   = \frac{1}{M}  \sum_{j=1}^M \left( \frac{\hat{\tau}_{I_{pq}}}{\tau_{V_{pq}}} \right)_j,\\
 \hat{\tau}_{I_{qp}}    &= \frac{1}{M}  \sum_{j=1}^M \left( \frac{\hat{\tau}_{I_{qp}}}{\tau_{V_{pq}}} \right)_j .\\   
         \end{aligned}
\end{equation}

One limitation of Algorithm \ref{alg:SLIC_Algo_1} is that it will give good estimates only when a RQM is present at one end of the branch. However, since RQMs are very expensive, they cannot be placed at one end of every branch of the connected tree. 
A strategy to solve the SLIC problem for multiple branches using a single RQM is described in the next section.


\section{Solving the System-wide SLIC Problem}
\label{SLIC Modeling - HV Network Extension}



\subsection{Extending IB-SLIC to Multiple Branches}
When a substation is selected for PMU placement, it is usually ensured that 
all the branches coming out of that substation are monitored by PMUs \cite{pal2016pmu}. In such a scenario, one often ends up getting multiple looks at the bus voltage
from PMUs that have different instrumentation infrastructure. This is very valuable in the context of the SLIC problem because \textit{the true voltage of the bus remains the same irrespective of which PMU is used to observe it}. 
We first explain the role of multiple independent observations of the bus voltage
in doing SLIC for a branch that is next 
to a branch that has an RQM. 
Then, we generalize it to all branches of the network for which SLIC must be performed.
Note that throughout this analysis, we assume a connected tree, $\mathcal{L}$, is present via which one can reach other 
branches from the branch that has the RQM.


Consider a branch $q$-$r$ that is next to the branch $p$-$q$ which has an RQM at the $p$-end.
Let the true voltage of bus $q$ be $V_q^*$. Since the bus $q$ is at the `to' end of branch $p$-$q$ and also `from' end of branch $q$-$r$, by using \eqref{eqn:Relation after IT - true_Ks}, $V_q^*$ can be expressed in two different ways:
\begin{equation}
\label{eqn:Vq-redundant}
         \begin{aligned}
        V_q^* = \tau_{V_{qp}} (V_{qp} - V_{qp_e})  \\
        V_q^* = \tau_{V_{qr}} (V_{qr} - V_{qr_e}).  \\
         \end{aligned}
\end{equation}

A variable $\rho_q$ is now defined as shown below.
\begin{equation}
\label{eqn:gamma_definition}
         \begin{aligned}
        \rho_q = \frac{(V_{qp} - V_{qp_e})}{(V_{qr} - V_{qr_e})} = \frac{\tau_{V_{qr}}}{\tau_{V_{qp}}}. \\
         \end{aligned}
\end{equation}

It is clear from \eqref{eqn:gamma_definition} that $\rho_q$ is the ratio of the correction factors of the VTs located on either side of bus $q$.
An estimate of $\rho_q$ (denoted by $\hat{\rho}_{q}$) can be obtained offline from $N$ noisy phasor measurements as shown below:
\begin{equation}
\label{eqn:gamma_estimate}
         \begin{aligned}
       \hat{\rho}_{q} = \frac{1}{N} \sum\limits_{j=1}^N \frac{V_{qp}(j)}{V_{qr}(j)}.
         \end{aligned}
\end{equation}

In \eqref{eqn:gamma_estimate}, we take advantage of the fact that noise added by the PMU device is small and has a zero-mean Gaussian distribution.
Now, $\hat{\rho}_{q}$ can be calculated using \eqref{eqn:gamma_estimate} $\forall{j}$, $j \in [1,M]$. Similarly, the correction factor ratios can also be estimated and saved $\forall{j}$, $j \in [1,M]$.
Lastly, by using $\hat{\rho}_{q_j}$ and the corresponding correction factor ratios, the correction factor estimate relation can be expressed as shown below.
\begin{equation}
\label{eqn:gamma_transfer_1_hop}
         \begin{aligned}
    \hat{\tau}_{V_{qr}}  &= \frac{1}{M} \sum_{j=1}^M \left( \hat{\rho}_{q_j} \left( \frac{\tau_{V_{qp}}}{\tau_{V_{pq}}} \right)_j  \right)\\
    \hat{\tau}_{V_{rq}}  &=    \frac{1}{M} \sum_{j=1}^M \left( \hat{\rho}_{q_j} \left( \frac{\tau_{V_{qp}}}{\tau_{V_{pq}}}  \right)_j \left(  \frac{\tau_{V_{rq}}}{\tau_{V_{qr}}} \right)_j  \right)\\
        \hat{\tau}_{I_{qr}}&=   \frac{1}{M} \sum_{j=1}^M \left( \hat{\rho}_{q_j} \left( \frac{\tau_{V_{qp}}}{\tau_{V_{pq}}}  \right)_j \left( \frac{\tau_{I_{qr}}}{\tau_{V_{qr}}}  \right)_j  \right)\\
     \hat{\tau}_{V_{rq}}   &=    \frac{1}{M} \sum_{j=1}^M \left(\hat{\rho}_{q_j} \left( \frac{\tau_{V_{qp}}}{\tau_{V_{pq}}} \right)_j \left( \frac{\tau_{I_{rq}}}{\tau_{V_{qr}}}  \right)_j  \right).\\
         \end{aligned}
\end{equation}

Note that apart from $\hat{\rho}_{q_j}$, the other terms in the RHS of \eqref{eqn:gamma_transfer_1_hop} are obtained using Algorithm \ref{alg:SLIC_Algo_1}.
It is also clear from \eqref{eqn:gamma_transfer_1_hop} that by using multiple independent voltage measurements
we can get a very good estimate of all the correction factors of the branch $q$-$r$ that is directly connected to the branch $p$-$q$ which has the RQM at the $p$-end.

Next, we generalize this property to any branch that connects to the RQM branch 
through other non-RQM branches. 
For ease of generalization, let the branches be named $ (1,2), (2,3), (3,4), \dots,  (u-1, u), (u, u+1)$. 
Furthermore, let the RQM be placed at the $1$-end of branch $(1,2)$, and the branch for which SLIC must be performed, be $(u,u+1)$, which is labeled as ``Current branch" in Fig. \ref{RQM_to_current_network_actiev_path}.
The black lines in the figure depict the connected tree, while the branches that do not belong to the connected tree are shown in grey.

\begin{figure}[H]
            \centering
            \includegraphics[width=0.45\textwidth]{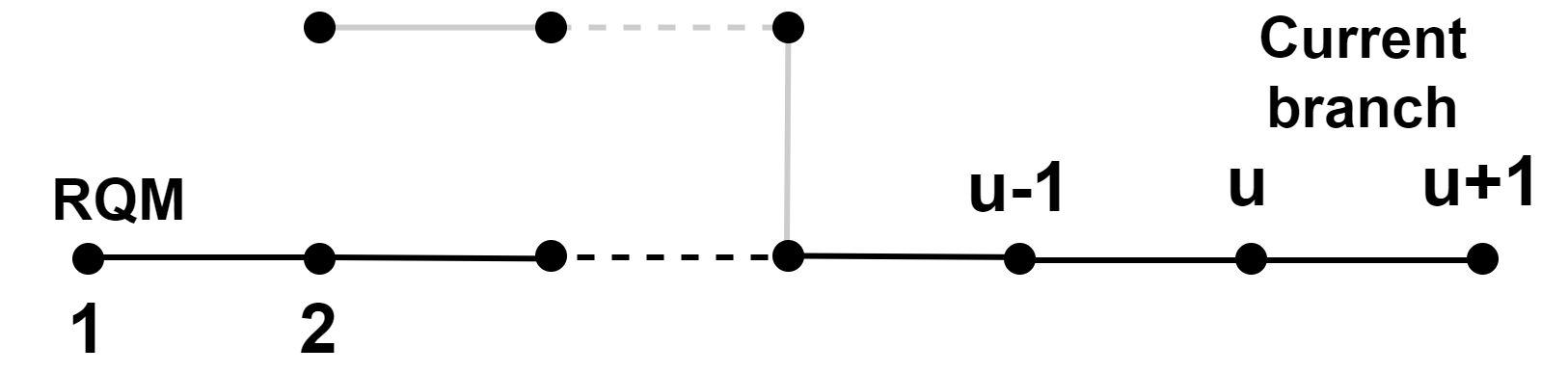}
            \caption{\textcolor{black}{HV branches from RQM branch to Current branch}}%
            \label{RQM_to_current_network_actiev_path}
\end{figure}

For the branch $(u,u+1)$, the relation of the correction factor ratios to the RQM located at the $1$-end of branch $(1,2)$
is computed in the following way.
First, compute $\Lambda$ as:
\begin{equation}
\label{Lambda_for_generic_kappa}
\begin{aligned}
\Lambda &= \frac{1}{M} \sum_{j=1}^M \left( \hat{\rho}_{u_j} \frac{\hat{\tau}_{V_{(u,u-1)}}}{\tau_{V_{(u-1,u)}}}   \hat{\rho}_{(u-1)} \dots  \frac{\hat{\tau}_{V_{(3,2)}}}{\tau_{V_{(2,3)}}}   \hat{\rho}_2 \frac{\hat{\tau}_{V_{(2,1)}}}{\tau_{V_{(1,2)}}} \right)  \\
&= \frac{1}{M} \sum_{j=1}^M \left( \left( \prod_{k=2}^{(u)} \hat{\rho}_{k_j} \right) \left( \prod_{(u, u+1)=(1,2)}^{((u-1), u))}  \frac{\tau_{V_{(u+1,u)}}}{\tau_{V_{(u,u+1)}}}  \right) \right).
\end{aligned}
\end{equation}

Then, the correction factors for $(u,u+1)$ is computed using $\Lambda$ and the correction factor ratios of every line segment in between $(u,u+1)$
and the RQM branch as shown below.
\begin{equation}
\label{Generic_kappa_Estimation_HV_network}
\begin{aligned}
&\hat{\tau}_{V_{(u,u+1)}} = \Lambda , \:\:
&\hat{\tau}_{V_{(u+1,u)}} = \Lambda  \: \frac{\hat{\tau}_{V_{(u+1,u)}}}{\tau_{V_{(u,u+1)}}} \\ 
&\hat{\tau}_{I_{(u,u+1)}} = \Lambda \: \frac{\hat{\tau}_{I_{(u,u+1)}}}{\tau_{V_{(u,u+1)}}}, \:\: 
&\hat{\tau}_{I_{(u+1,u)}} = \Lambda \: \frac{\hat{\tau}_{I_{(u+1,u)}}}{\tau_{V_{(u,u+1)}}}  . \\
\end{aligned}
\end{equation}

From \eqref{Generic_kappa_Estimation_HV_network}, it can be realized that by leveraging the multiple independent observations
of the bus voltages, the formulation described above is able to estimate $\mathcal{T}_{pq}, \forall (p,q) \in \mathcal{L}$
using an RQM placed at one end of
a branch.

In summary, using IB-SLIC (Algorithm \ref{alg:SLIC_Algo_1}) and multiple independent voltage measurements, we can estimate line parameters and calibrate ITs of the entire connected tree. 
The only requirements for the proposed approach are the presence of PMUs at both ends of the lines, and one RQM at the end of one of the lines. 
To reach the other lines from the line that has the RQM, we use depth-first search (DFS). 
The overall implementation of the proposed formulation for system-wide SLIC 
is described in Algorithm \ref{alg:SCALP_Algo_2}, where $\mathcal{B}$, $V_{\mathcal{L}}$, and $I_{\mathcal{L}}$ denote the buses, voltage phasor measurements, and current phasor measurements, respectively, of the connected tree ($\mathcal{L}$).

\begin{algorithm}
\caption{System-wide-SLIC (SW-SLIC)}\label{alg:SCALP_Algo_2}
\begin{algorithmic}[1]
\State \textbf{Input: } $\mathcal{B}, \mathcal{L}, V_{\mathcal{L}}, I_{\mathcal{L}}, \eta_{\dagger} $
\State \textbf{Output: } $\hat{\eta}, \hat{\Gamma}, \hat{\mathcal{T}}$
\Procedure{Path Finder}{$node, goal, path, visited$}
    \If {$node = goal$}
        \State $path.append(node)$
        \State \textbf{return} $path$
    \EndIf
    \State $visited[node] \gets True$
    \State $path.append(node)$
    \ForAll {$neighbor \in neighbors(node)$}
        \If {$visited[neighbor] = False$}
            \State result $\gets$ \Call{DFS}{$neighbor, goal, path, visited$}
            \If {$result \neq None$}
                \State \textbf{return} $result$
            \EndIf
        \EndIf
    \EndFor
    \State $path.pop()$  
    \State \textbf{return} $None$
\EndProcedure

\Procedure{SW-SLIC}{$\mathcal{B}, \mathcal{L}, V_{\mathcal{L}}, I_{\mathcal{L}}, \eta_{\dagger} $}
    \ForAll {$bus \in \mathcal{B}$}
        \State $path $ = [], $visited$ = []
        \State $\mathrm{PT}$ = \Call{Path Finder}{RQM, $bus$, $path$, $visited$}
        \State $[\hat{\eta}_{pq},  \hat{\Gamma}_{pq}]$  = IB-SLIC($p, q$, $V_{pq}, V_{qp}$, $I_{pq}$, $I_{qp}$, $\eta_{\dagger}$) 
        \State Calculate $\Lambda$ using  \eqref{Lambda_for_generic_kappa}, $\hat{\Gamma}_{pq}$, and $\mathrm{PT}$
        \State Estimate $\mathcal{T}_{pq}$ using \eqref{Generic_kappa_Estimation_HV_network}, $\Lambda$, and $\hat{\Gamma}_{pq}$ 
    \EndFor
    \State \textbf{return} $\hat{\eta}, \hat{\Gamma}, \hat{\mathcal{T}}$
\EndProcedure
\end{algorithmic}
\end{algorithm}


Algorithm \ref{alg:SCALP_Algo_2} comprises two main components. The initial component, referred to as the PATH FINDER function, is predicated on the DFS algorithm. This function can determine the connected tree
from the RQM to any specified bus when initialized with the RQM's position as its primary argument, while both the `path' and `visited' parameters are set to null. Alongside Algorithm \ref{alg:SLIC_Algo_1}, the PATH FINDER is instrumental in delineating the SW-SLIC process. This process, starting from the RQM bus, systematically traverses to each bus, identifies the next branch to perform the analysis, and employs the PATH FINDER to determine the connected tree
from the RQM to the current bus. Subsequently, the IB-SLIC is performed to derive the branch-specific parameters $\hat{\eta}_{pq}$ and $\hat{\Gamma}_{pq}$. The final stage involves estimating the correction factor, utilizing the connected tree
sourced from the PATH FINDER, the $\hat{\Gamma}_{pq}$ obtained through IB-SLIC, and \eqref{Lambda_for_generic_kappa}-\eqref{Generic_kappa_Estimation_HV_network}. This methodology is iteratively applied to each bus via a $\mathrm{for~loop}$
to ensure system-wide SLIC is performed $\forall (p,q) \in \mathcal{L}$.

\subsection{Optimal Location for the RQM}
\label{OptLoc_new_RQM_SLIC}

The presence of an RQM is crucial for the implementation of the proposed solution for the SLIC problem. 
If a connected tree already has an RQM, then one can carry out the procedure described in Algorithm \ref{alg:SCALP_Algo_2} in the previous sub-section. 
However, if a connected tree does not have an RQM, then one must first determine a suitable location for placing it.
This suitable location is found by the proposed new RQM placement algorithm (Algorithm \ref{alg:SCALP_Algo_3}), which is formulated on top of the SW-SLIC algorithm.
Note that Algorithm \ref{alg:SCALP_Algo_3} only helps satisfy
the needs of the SLIC problem and does not consider other benefits that an RQM might provide. 

In Algorithm \ref{alg:SCALP_Algo_3}, the SW-SLIC is repeated for all possible combinations of RQM locations. The first $\mathrm{for~loop}$ conducts the SW-SLIC keeping the RQM at the `from' end bus of all the branches (defined as $\mathcal{B}_f$) in $\mathcal{L}$. The second $\mathrm{for~loop}$ repeats the SW-SLIC keeping RQM at the `to' end of all the branches (defined as $\mathcal{B}_t$).
Both loops save the performance index $\mu_{\text{MARE}}$ in a variable, $\mathcal{E}$, where the subscript $\mathrm{MARE}$ indicates \textit{mean absolute relative error}.
Particularly, $\mu_{\text{MARE}}$ represents the average value calculated from a large number of Monte Carlo (MC) simulations, each of which computes the $\mathrm{MARE}$ value
for both the line parameters as well as the IT correction factors.
The location with the minimum value of $\mathcal{E}$ is chosen for RQM placement.

\begin{algorithm}
\caption{New RQM Placement}\label{alg:SCALP_Algo_3}
\begin{algorithmic}[1]
\State \textbf{Input: } $\mathcal{B}, \mathcal{L}, \mathcal{B}_f, \mathcal{B}_t, V_{\mathcal{L}}, I_{\mathcal{L}}, \eta_{\dagger}$
\State \textbf{Output: } OptLoc
\State Initialize $\mathcal{E} = \vec{0}$
\For{$k \in \mathcal{B}_f$}
\State $[\hat{\eta},  \hat{\Gamma}]$  = SW-SLIC( $\mathcal{B}_f, \mathcal{L}, V_{\mathcal{L}}, I_{\mathcal{L}}, \eta_{\dagger}$) 
\State Calculate  $\mu_{\text{MARE}}$
\State $\mathcal{E} (k) =\mu_{\text{MARE}}$
\EndFor
\For{$k \in \mathcal{B}_t$}
\State $[\hat{\eta},  \hat{\Gamma}]$  = SW-SLIC( $\mathcal{B}_t, \mathcal{L}, V_{\mathcal{L}}, I_{\mathcal{L}}, \eta_{\dagger}$) 
\State Calculate  $\mu_{\text{MARE}}$
\State $\mathcal{E} (k+|\mathcal{L}|) =\mu_{\text{MARE}}$
\EndFor
\State OptLoc = $\arg\underset{i }\min~ \mathcal{E} (i)$
\end{algorithmic}
\end{algorithm}

Note that Algorithm \ref{alg:SCALP_Algo_3} is built upon the SW-SLIC algorithm which makes use of DFS. 
Therefore, a \textit{centrally} located RQM is expected to improve the accuracy of IT correction factors as it will ensure
traversal across minimum number of hops to reach all the buses in $\mathcal{L}$ from the RQM bus.
This intuition is corroborated in the next section, in which 
we evaluate the performance of the proposed approach for solving the system-wide SLIC problem for different power system conditions.


\section{Results}
\label{Results}

The proposed methodology can be applied to any system as long as a \textit{connected tree} is present.
However, since PMUs are usually placed on the highest voltage (HV) buses first \cite{varghese2023timesynchronized}, we selected the 345 kV network of the IEEE 118-bus system, which has 11 buses 
connected by 10 branches.
For the U.S. power utility, the analysis was carried out on four buses connected by three branches.
All simulations were performed 
on a computer having 64 GB RAM with an Intel 11th Gen Intel(R) Core(TM) i7 processor @3.00 GHz.

For the 118-bus system, we simulated the morning-load pickup \cite{gao2012dynamic1} in MATPOWER.
For diverse loading conditions, we calculated active and reactive power injections, and then solved the power flow to obtain the corresponding voltage and current phasor measurements. 
This procedure was repeated after introducing bounded perturbations ($\leq\!\!30\%$) in the line parameters.
The measurements thus generated were the true phasor values.
Next, REs of normal quality ITs (accuracy class of 0.6 \cite{IEEE_C57_13_2016_std_for_ITs}) and RQMs (accuracy class of 0.15 \cite{IEEE_C57_13_2016_std_for_ITs}) were multiplied with the true phasors to mimic the outputs of non-ideal ITs.
Finally, a zero-mean Gaussian noise with a standard deviation of 0.03\% was added to replicate noisy PMU outputs from uncalibrated ITs.
The resulting dataset was fed into the proposed approach for solving the SLIC problem.
The proposed approach first employed TLS to solve \eqref{eqn:SLIC_real_abstraction} and obtain $\hat{\theta}$.
Next, $\hat{\eta}$ and $\hat{\rho}$ for every branch in the network are obtained using IB-SLIC (Algorithm \ref{alg:SLIC_Algo_1}). 
Finally, SW-SLIC (Algorithm \ref{alg:SCALP_Algo_2}) is utilized to estimate all the line parameters and IT correction factors in the connected tree.
To evaluate the performance of the proposed approach, the $\mathrm{ARE}$ index is chosen as the performance metric.
For variables that are in the complex domain ($\Gamma$ and $\mathcal{T}$), absolute value of $\mathrm{ARE}$ is 
displayed for sake of brevity (the inferences drawn from the angle comparisons are similar).
To compare statistics across $M$ MC runs, the mean of $\mathrm{ARE}$, denoted by $\mathrm{MARE}$, is computed. 
The results obtained are summarized below.

\vspace{-1em}

\subsection{Validation of Proposed Solution to SLIC Problem}
\label{SLIC - Performance Analysis - HV Network results}

For the analysis done in this sub-section, the RQM is assumed to be at $10$-end of branch $9$-$10$.
The proposed approach is repeated $M=1000$ times, and the number of bins are set to $61$ (= $-0.3$ to $+0.3$ in steps of $0.01$).
The $\mathrm{ARE}$ of the estimates is saved after every MC simulation. 
The $\mathrm{MARE}$ across all the MC simulations is displayed in Tables \ref{LPE_ARE_table_non_optimal_RQM}-\ref{CF_ARE_table_non_optimal_RQM}.

Table \ref{LPE_ARE_table_non_optimal_RQM} presents $\mathrm{MARE}$s of $r$, $x$, and $b$ for the 10 branches after applying
Algorithm \ref{alg:SCALP_Algo_2}. The generally small $\mathrm{MARE}$ values ($<1\%$) confirm the excellent estimation accuracy of the proposed approach.
Moreover, it is evident from the table  that despite variations in the operating conditions and the branches' distances from the RQM branch, the accuracy of the estimates remained sufficiently high.
For the two branches where the 
estimation error surpassed $1\%$ ($L_{65-68}$ and $L_{68-81}$), further analysis was conducted, which
revealed that 
bins adjacent to the actual bin were 
selected as the optimal bin on many occasions by Algorithm \ref{alg:SLIC_Algo_0}. This occurred due to a combination of two factors: (i) additive noise introduced by the PMU, and (ii) value of $W_{pq}$ across all $61$ bins lying within a very narrow range (numerically).
The combined effect was that even a small error in the TLS estimate resulted in 
a higher error in the line parameter estimates for these two branches.

\begin{table}[b]
\caption{\% $\mathrm{MARE}$s for Line Parameters}
\label{LPE_ARE_table_non_optimal_RQM}
\fontsize{9}{9}
\centering
\begin{tabular}{|l|c|c|c|c|c|}
\hline
Branch    & $L_{30-38}$ & $L_{38-65}$ & $L_{65-64}$ & $L_{64-63}$  & $L_{65-68}$ \\ \hline
$r_{pq}$  &    0.632   &   0.436   &   0.441         &  0.567       &   1.88      \\ \hline
$x_{pq}$  & 0.635    &  0.433    & 0.448           & 0.562        &    1.87     \\ \hline
$b_{pq}$  &  0.627    &  0.445    &  0.437          &  0.561       &   1.92      \\ \hline
Branch    & $L_{68-81}$ & $L_{30-26}$ & $L_{8-30}$ & $L_{8-9}$   & $L_{9-10}$  \\ \hline
$r_{pq}$  &  2.13    & 0.436          &   0.462    & 0.422     & 0.432      \\ \hline
$x_{pq}$  & 2.13     & 0.429         &   0.475    &   0.416   &   0.427    \\ \hline
$b_{pq}$  &  2.12    &  0.425         &     0.459  &   0.436   &  0.434     \\ \hline
\end{tabular}
\end{table}




The absolute $\mathrm{MARE}$s of the correction factor ratios obtained using the proposed approach 
are tabulated in Table \ref{CFR_ARE_table_non_optimal_RQM}.
One trend that we observe in this table is that the estimates of the VT correction factor ratios ($\frac{\hat{\tau}_{V_{qp}}}{\tau_{V_{pq}}}$) are consistently more accurate than the estimates of the two CT correction factor ratios ($\frac{\hat{\tau}_{I_{pq}}}{\tau_{V_{pq}}}$ and $\frac{\hat{\tau}_{I_{qp}}}{\tau_{V_{pq}}}$). This observation can be explained in the following way.
The value of $\widehat{W}_{pq}$ is directly obtained by solving the linear system of equations in \eqref{eqn:SLIC_real_abstraction}. Now, from \eqref{SLIC:theta_parameters_Expression}, it can be realized that the calculation of the VT correction factor ratios require only $\widehat{W}_{pq}$ and $\hat{\theta}$.
On the contrary, the CT correction factor ratios require the knowledge of line parameter estimates in addition to $\widehat{W}_{pq}$ and $\hat{\theta}$.
Since the line parameter estimates can have a small estimation error due to the quantization procedure (as observed in Table  \ref{LPE_ARE_table_non_optimal_RQM}), this error is carried over to the CT correction factor ratio estimation.
This is the reason why the VT correction factor ratios have relatively better estimation accuracies as opposed to the CT correction factor ratios.

\begin{table}[b]
\caption{\% 
Absolute 
$\mathrm{MARE}$s for Correction Factor Ratios}
\label{CFR_ARE_table_non_optimal_RQM}
\fontsize{9}{9}
\centering
\begin{tabular}{|l|c|c|c|c|c|}
\hline
Branch                                       & $L_{30-38}$   & $L_{38-65}$   & $L_{65-64}$     & $L_{64-63}$  & $L_{65-68}$ \\ \hline
$\frac{\hat{\tau}_{V_{qp}}}{\tau_{V_{pq}}}$  &    0.09       &   0.06        &   0.4           &  0.05        &   0.03      \\ \hline
$\frac{\hat{\tau}_{I_{pq}}}{\tau_{V_{pq}}}$  & 1.40          &  0.30         & 0.53            & 1.04         &    4.41    \\ \hline
$\frac{\hat{\tau}_{I_{qp}}}{\tau_{V_{pq}}}$  &  1.42         &  0.35         &  0.53           &  1.02        &   4.44      \\ \hline
Branch                                       & $L_{68-81}$   & $L_{30-26}$   & $L_{8-30}$      & $L_{8-9}$    & $L_{9-10}$  \\ \hline
$\frac{\hat{\tau}_{V_{qp}}}{\tau_{V_{pq}}}$  &  0.09         & 0.06          &   0.06          & 0.05         & 0.05     \\ \hline
$\frac{\hat{\tau}_{I_{pq}}}{\tau_{V_{pq}}}$  & 4.61          & 0.29          &   0.91          &   0.36       &   0.34    \\ \hline
$\frac{\hat{\tau}_{I_{qp}}}{\tau_{V_{pq}}}$  &  4.69         &  0.33         &     0.92        &   0.38       &  0.36     \\ \hline
\end{tabular}
\end{table}

Table \ref{CF_ARE_table_non_optimal_RQM} shows the absolute $\mathrm{MARE}$s of the correction factors for the 10 branches.
The same observation (as Table \ref{CFR_ARE_table_non_optimal_RQM}) can be made for the correction factors (i.e., the estimates of the VT correction factors are better than the estimates of the CT correction factors). 
Additionally, it can be noticed that for most of the lines, the results are very good. The slightly higher 
$\mathrm{MARE}$ results are for those two branches, whose line parameter values make it challenging to select the optimal bin
(due to the two factors mentioned earlier).

\begin{table}
\caption{\% 
Absolute 
$\mathrm{MARE}$s for Correction Factors}
\label{CF_ARE_table_non_optimal_RQM}
\fontsize{9}{9}
\centering
\begin{tabular}{|l|c|c|c|c|c|}
\hline
Branch                 & $L_{30-38}$ & $L_{38-65}$     & $L_{65-64}$     & $L_{64-63}$  & $L_{65-68}$ \\ \hline
$\hat{\tau}_{V_{pq}}$  &    0.156    &   0.188         &   0.205         &  0.211       &   0.205      \\ \hline
$\hat{\tau}_{V_{qp}}$  & 0.188       &  0.205          & 0.212           & 0.215        &    0.207     \\ \hline
$\hat{\tau}_{I_{pq}}$  &  1.394      &  0.344          &  0.569          &  1.074       &   4.39     \\ \hline
$\hat{\tau}_{I_{qp}}$  &  1.407      &  0.393          &  0.573          &  1.061       &   4.42      \\ \hline
Branch                 & $L_{68-81}$ & $L_{30-26}$     & $L_{8-30}$      & $L_{8-9}$    & $L_{9-10}$  \\ \hline
$\hat{\tau}_{V_{pq}}$  &  0.207      & 0.151           &   0.144         & 0.129        & 0.125      \\ \hline
$\hat{\tau}_{V_{qp}}$  & 0.223       & 0.169           &   0.156         &   0.143      &   -- --
\\ \hline
$\hat{\tau}_{I_{pq}}$  &  4.61       &  0.333          &     0.924       &   0.379      &  0.356     \\ \hline
$\hat{\tau}_{I_{qp}}$  &  4.69       &  0.366          &     0.935       &   0.393      &  0.374     \\ \hline
\end{tabular}
\end{table}

\vspace{-1em}

\subsection{Optimal RQM Placement Results}
\label{RQM_placement_empirical_Results}
In the previous analysis, the RQM was pre-placed at $10$-end of branch $9$-$10$. 
However, if there is no RQM to begin with, then the optimal location for placing the RQM to solve the SLIC problem can be found using Algorithm \ref{alg:SCALP_Algo_3}.
Since the connected tree identified for the 118-bus system has 10 branches, 20 different locations are possible where the RQM can be placed.
To compare the accuracies across these locations, we compute the $\mu_{\mathrm{MARE}}$ index as explained in Section \ref{OptLoc_new_RQM_SLIC}. 
The results are summarized in Table \ref{RQM_placement_using_CF_ARE}, in which the $\mu_{\mathrm{MARE}}$ is calculated by systematically placing the lone RQM at either the `from' end or the `to' end of every branch in the connected tree.

\begin{table}[H]
\caption{Optimal RQM Placement using $\mu_{\mathrm{MARE}}$ Index}
\label{RQM_placement_using_CF_ARE}
\fontsize{9}{9}
\centering
\begin{tabular}{|l|c|c|c|c|c|}
\hline
Branch   & $L_{8-30}$  & $L_{30-38}$ & $L_{38-65}$ & $L_{65-64}$ & $L_{64-63}$ \\ \hline
From end & 0.793       & 0.782      & 0.793       & 0.787       & 0.801       \\ \hline
To end   & 0.786       & 0.791       & 0.787      & 0.793       & 0.806       \\ \hline
Branch   & $L_{65-68}$ & $L_{68-81}$ & $L_{30-26}$ & $L_{8-9}$   & $L_{9-10}$  \\ \hline
From end & 0.787       & 0.806       & 0.788       & 0.799       & 0.799       \\ \hline
To end   & 0.792       & 0.804       & 0.787       & 0.801       & 0.804       \\ \hline
\end{tabular}
\end{table}

It can be observed from Table \ref{RQM_placement_using_CF_ARE} that an RQM located at $30$-end of branch $30$-$38$ gives the best performance (lowest $\mu_{\mathrm{MARE}}$ index).
Moreover, centrally located buses have lower indices than the buses in the periphery of the connected tree. 
This is in-line with our intuition that a centrally located node is suitable for RQM placement (see last paragraph of Section \ref{OptLoc_new_RQM_SLIC}). This result also implies that for a large connected tree in which it may be computationally burdensome
to compute $\mu_{\mathrm{MARE}}$ for every possible location, the bus identified using a centrality measure, such as betweenness centrality,
would be a good choice for placing the RQM for the SLIC problem.

Finally, we re-do the analysis done in Section \ref{SLIC - Performance Analysis - HV Network results} by placing the RQM at $30$-end of branch $30$-$38$. The results are shown in Tables \ref{LPE_ARE_table_optimal_RQMp}-\ref{CF_ARE_table_optimal_RQMp}.
The $\mathrm{MARE}$s of the line parameter estimates decreased from 0.78\% to 0.73\% on average (compare Tables \ref{LPE_ARE_table_non_optimal_RQM} and \ref{LPE_ARE_table_optimal_RQMp}).
Similarly, the average $\mathrm{MARE}$ of the correction factors improved from 0.81\% to 0.78\% (compare Tables \ref{CF_ARE_table_non_optimal_RQM} and \ref{CF_ARE_table_optimal_RQMp}).
In summary, it is realized that if a utility has the choice, then it should place the RQM at the location determined by Algorithm \ref{alg:SCALP_Algo_3} for solving the SLIC problem.

\begin{table}[H]
\caption{\% $\mathrm{MARE}$s for Line Parameters}
\label{LPE_ARE_table_optimal_RQMp}
\fontsize{9}{9}
\centering
\begin{tabular}{|l|c|c|c|c|c|}
\hline
Branch    & $L_{30-38}$ & $L_{38-65}$ & $L_{65-64}$ & $L_{64-63}$  & $L_{65-68}$ \\ \hline
$r_{pq}$  &    0.673   &   0.366   &   0.386         &  0.534       &   1.64      \\ \hline
$x_{pq}$  & 0.675    &  0.365    & 0.383           & 0.526        &    1.65     \\ \hline
$b_{pq}$  &  0.671    &  0.371    &  0.387          &  0.542       &   1.64      \\ \hline
Branch    & $L_{68-81}$ & $L_{30-26}$ & $L_{8-30}$ & $L_{8-9}$   & $L_{9-10}$  \\ \hline
$r_{pq}$  &  2.10    & 0.361          &   0.463    & 0.388     & 0.384      \\ \hline
$x_{pq}$  & 2.11     & 0.356         &   0.459    &   0.385   &   0.381    \\ \hline
$b_{pq}$  &  2.11    &  0.368         &     0.468  &   0.379   &  0.375     \\ \hline
\end{tabular}
\end{table}

\begin{table}[H]
\caption{\% Absolute $\mathrm{MARE}$s for Correction Factors}
\label{CF_ARE_table_optimal_RQMp}
\fontsize{9}{9}
\centering
\begin{tabular}{|l|c|c|c|c|c|}
\hline
Branch                 & $L_{30-38}$ & $L_{38-65}$     & $L_{65-64}$     & $L_{64-63}$  & $L_{65-68}$ \\ \hline
$\hat{\tau}_{V_{pq}}$  &    -- --
&   0.144         &   0.159         &  0.163       &   0.159      \\ \hline
$\hat{\tau}_{V_{qp}}$  & 0.145       &  0.160          & 0.164           & 0.167        &    0.161     \\ \hline
$\hat{\tau}_{I_{pq}}$  &  1.371      &  0.295          &  0.575          &  0.964       &   4.33     \\ \hline
$\hat{\tau}_{I_{qp}}$  &  1.373      &  0.337          &  0.578          &  0.967       &   4.36      \\ \hline
Branch                 & $L_{68-81}$ & $L_{30-26}$     & $L_{8-30}$      & $L_{8-9}$    & $L_{9-10}$  \\ \hline
$\hat{\tau}_{V_{pq}}$  &  0.161      & 0.114           &   0.115         & 0.119        & 0.133      \\ \hline
$\hat{\tau}_{V_{qp}}$  & 0.193       & 0.127           &   0.120         &   0.133      &   
-- --
\\ \hline
$\hat{\tau}_{I_{pq}}$  &  4.84       &  0.342          &     0.893       &   0.365      &  0.390     \\ \hline
$\hat{\tau}_{I_{qp}}$  &  4.88       &  0.374          &     0.920       &   0.383      &  0.407     \\ \hline
\end{tabular}
\end{table}


\subsection{Sensitivity Analysis}
In this sub-section, we investigate sensitivity of the proposed approach to varying levels of
PMU noise and 
IT accuracy class to understand how they impact its performance.

\subsubsection{Additive PMU Noise}
\label{SLIC - Sensitivity Analysis - Additive Noise}
While keeping all other parameters the same as Section \ref{RQM_placement_empirical_Results}, the standard deviation (denoted by $\sigma$) of the Gaussian noise of the PMU was changed from $0.01\%$ to $0.03\%$ to $0.05\%$.
Fig. \ref{Imapct_addi_noise_scalp_r_3} shows the $\mathrm{MARE}$ of the resistance estimates obtained after $M=1000$ MC simulations, each of which was conducted 
for a specific value of $\sigma$.
Very similar results were obtained for the reactance and susceptance estimates as well, which is why the following description focuses on only the resistance estimates (provided in Fig. \ref{Imapct_addi_noise_scalp_r_3}). 

\begin{figure}[H]
    \centering
    \includegraphics[width=0.36\textwidth, trim={0.5cm 0.5cm 0.5cm 0.5cm}]{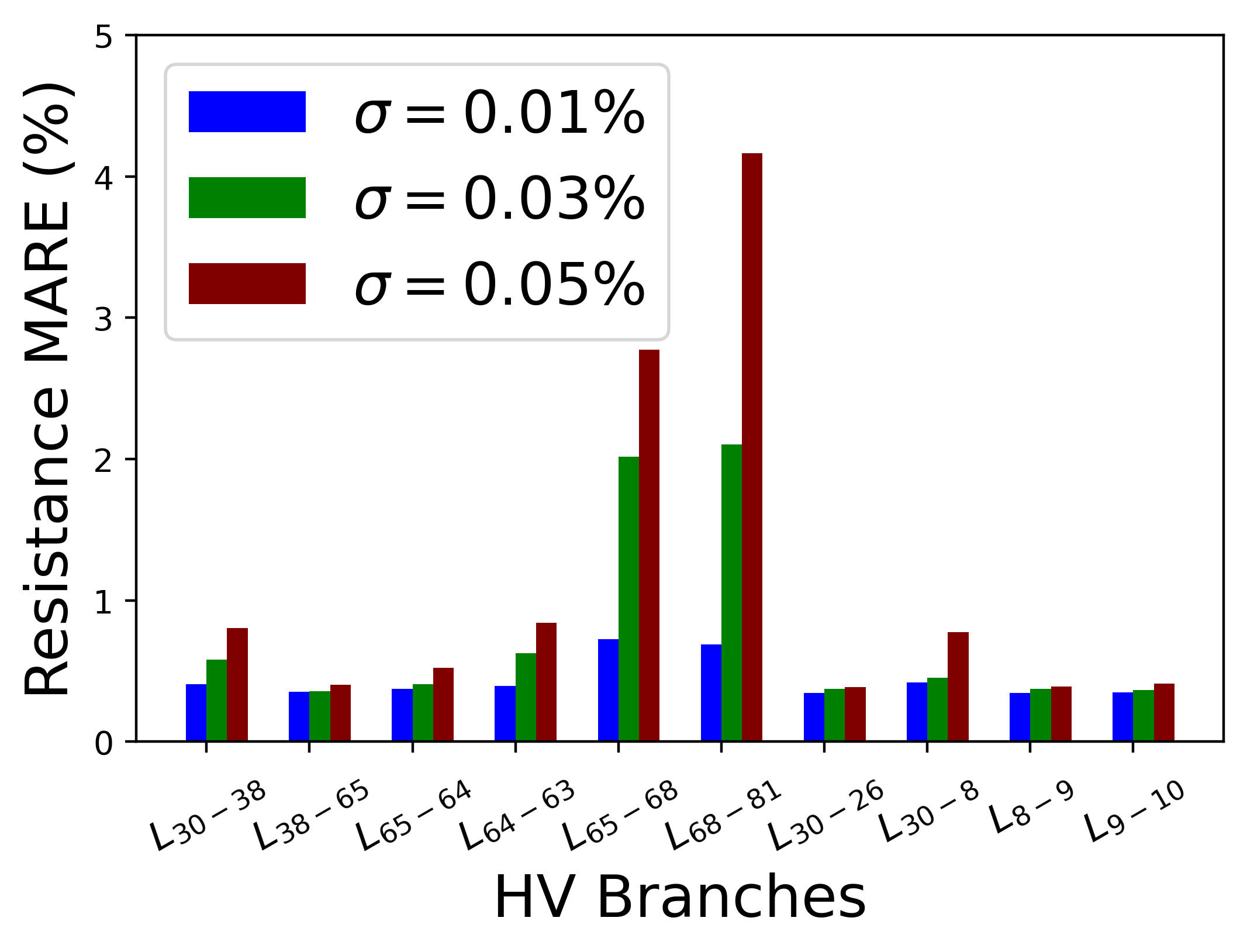}
    \caption{Impact of additive noise on resistance estimate}
    \label{Imapct_addi_noise_scalp_r_3}
\end{figure}

In line with our expectation, as $\sigma$ increases, the $\mathrm{MARE}$s increase as well. 
The relatively higher errors for $L_{65-68}$ and $L_{68-81}$ 
is due to the two factors mentioned in Section \ref{SLIC - Performance Analysis - HV Network results} 
(a combination of additive noise and very small variation in the $W_{pq}$ values across all the bins).
Particularly, the additive noise impacts the general performance of the proposed approach in two ways: 
The first is when TLS is done to solve \eqref{eqn:SLIC_real_abstraction} to obtain $\hat{\theta}$. In presence of low additive noise, a very accurate estimate of $\theta$ can be obtained. However, as the amount of noise increases, the difference between $\hat{\theta}$ and the true value of $\theta$, namely $\theta^*$, increases. This lowering of the TLS accuracy impacts the subsequent estimations.
The second way in which the additive noise affects the proposed approach is in the calculation of the variable $\rho$. This variable was used to propagate the accuracy of RQMs across multiple branches. However, $\rho$ is calculated from the measurements while suppressing
the additive noise (see \eqref{eqn:gamma_definition}). Therefore, a higher additive noise would result in a less accurate estimation of $\rho$, which in turn would impact the accuracy of the subsequent estimates.

\subsubsection{IT Accuracy Class}
In this sensitivity study, the performance of the proposed approach for SLIC is analyzed for three different classes of regular ITs (namely, 0.3, 0.6, and 1.2)
\cite{IEEE_C57_13_2016_std_for_ITs}.
The accuracy of the RQM is not unchanged.
That is, every parameter other than the RE limit of the regular ITs is the same as that in Section \ref{RQM_placement_empirical_Results}. 
The average $\mathrm{MARE}$s of the line parameter estimates and the average absolute $\mathrm{MARE}$s of the IT correction factors - both average being computed across all ten branches - are displayed in Fig. 
\ref{Impact_IT_class}.

\begin{figure}[ht]
    \centering
    \begin{subfigure}[b]{0.45\linewidth} 
        \includegraphics[width=\linewidth]{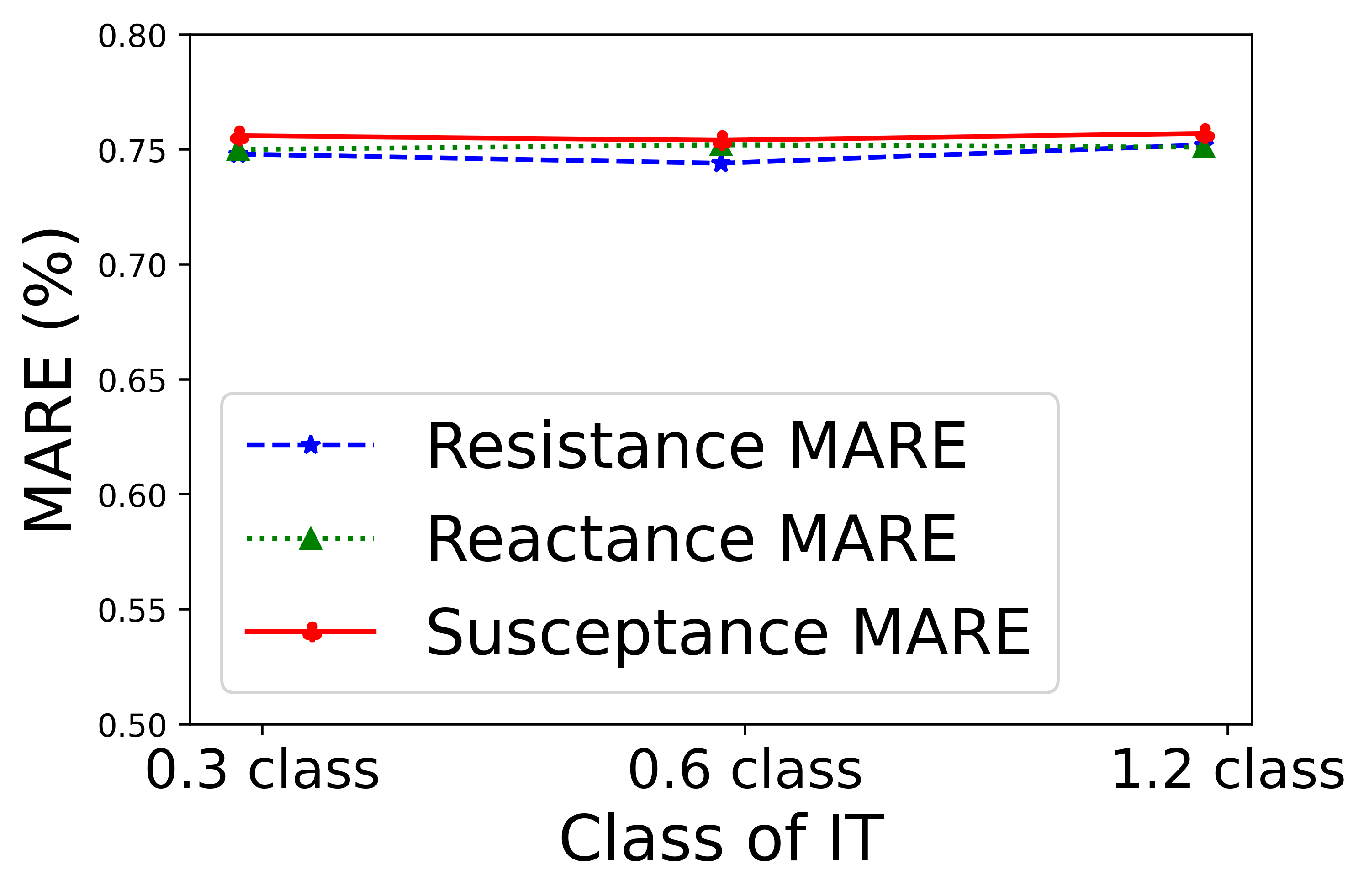}
        \caption{Line parameter estimates}
        \label{Impact_IT_class_LPE}
    \end{subfigure}
    \begin{subfigure}[b]{0.45\linewidth} 
        \includegraphics[width=\linewidth]{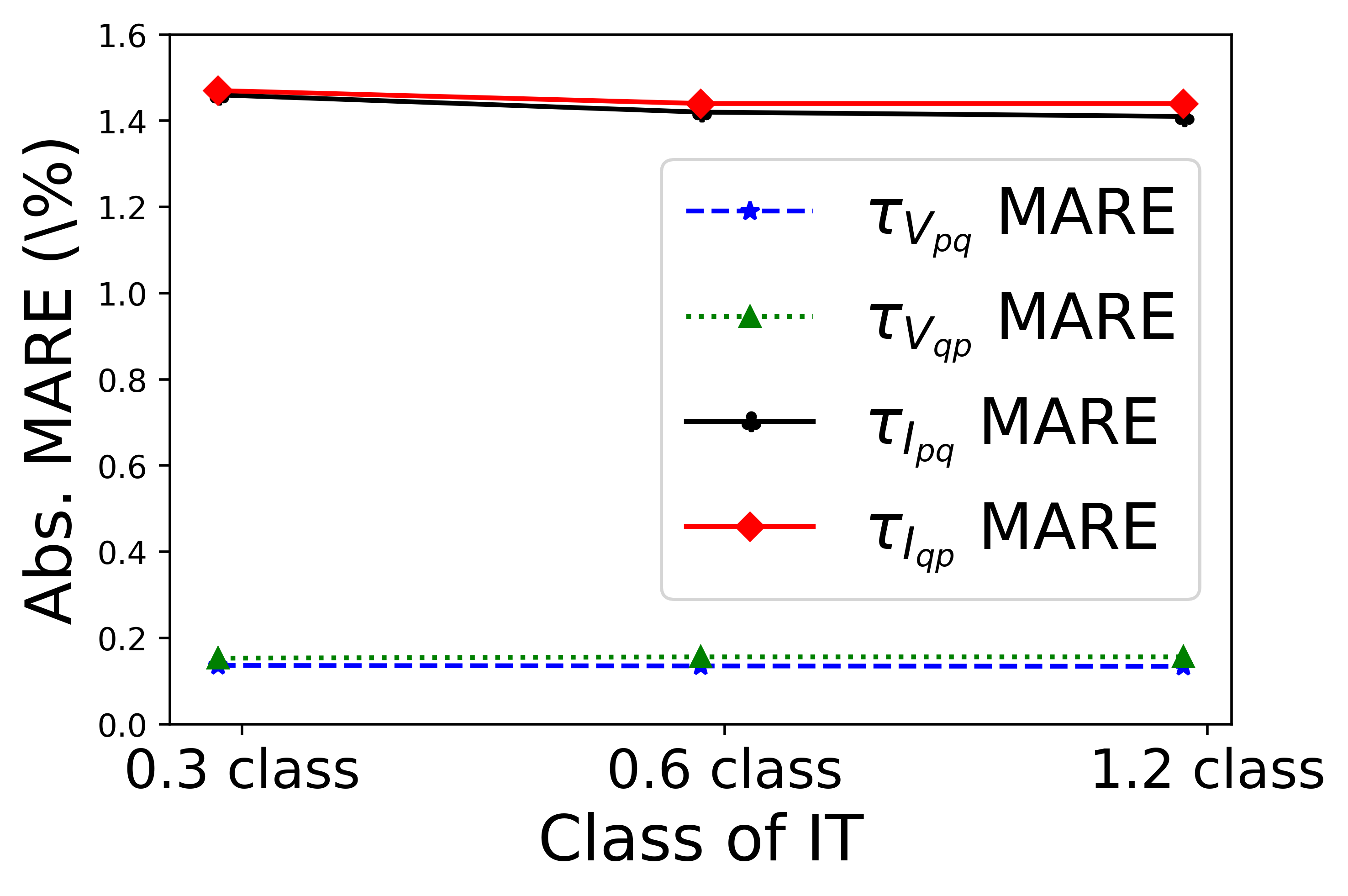}
        \caption{Correction factor estimates}
        \label{Impact_IT_class_CF}
    \end{subfigure}
    \caption{Impact of different IT classes on proposed solution}
    \label{Impact_IT_class}
\end{figure}

It can be observed from Fig. \ref{Impact_IT_class}
that the different accuracy classes of regular ITs have minimal impact on the performance of the proposed approach for SLIC. 
This happens because RE of the regular ITs are treated as an unknown parameter during the estimation process, and compensated subsequently by finding the appropriate correction factor.
Thus, the proposed approach is fundamentally superior to prior LPE approaches (such as \cite{mishra2015kalman,mansani2018estimation}) that estimated the line parameters without considering the effects of IT errors.

\subsection{Comparison with a Similar State-of-the-Art Approach}

A recent study by Wang et al. \cite{wang2019transmission} also addressed the SLIC problem in the context of transmission systems in which PMUs were only placed on the HV buses of the network; i.e., the lines for which they performed SLIC matched the definition of connected tree of our paper.
They placed their RQMs at $81$-end of branch $68$-$81$ of the 118-bus system, and proved analytically that the additive PMU noise has zero-mean Gaussian distribution in Cartesian coordinates.
However, \cite{wang2019transmission} required EOVTs as well as MOCTs at that location (i.e., RQMs for both voltage and current ITs) in comparison to just EOVT for the proposed approach. Moreover, their PMU noise had a $\sigma=0.0003\%$.
In this sub-section, we compare the performance of the proposed approach with \cite{wang2019transmission} for LPE (Fig. \ref{Impact_IT_class_LPE_comparison_chenwang}) and IT correction factor estimation (Fig. \ref{Impact_IT_class_CFE_comparison_chenwang}), respectively. For fairness of comparison, we placed our EOVT (RQM) at $81$-end of branch $68$-$81$, and changed the standard deviation of the additive PMU noise to match the one used in \cite{wang2019transmission}.

\begin{figure}[ht]
    \centering
    \includegraphics[width=0.85\linewidth]{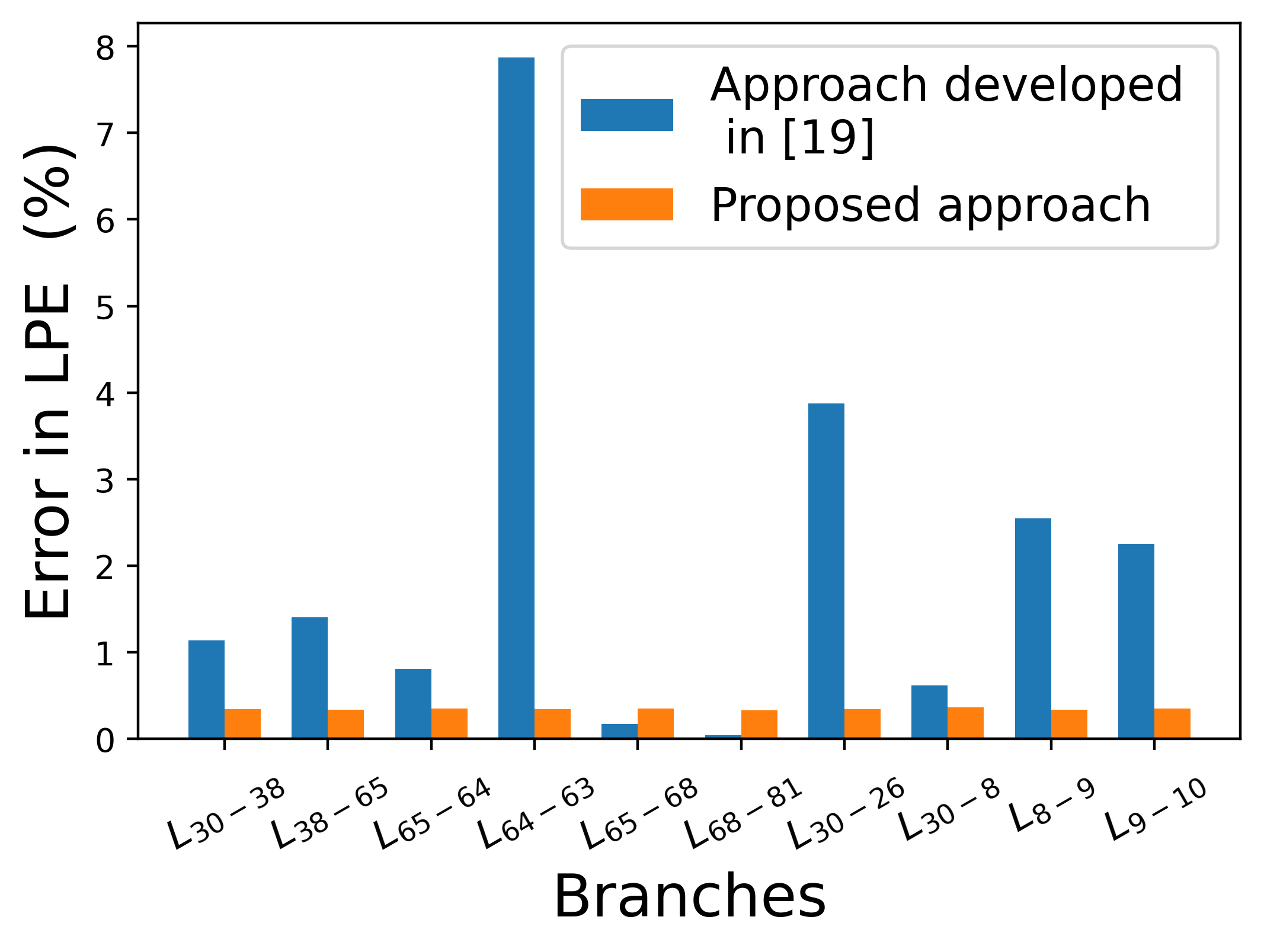}
    \vspace{-0.5em}
    \caption{Comparing line parameter estimates with \cite{wang2019transmission} }
\label{Impact_IT_class_LPE_comparison_chenwang}
\end{figure}
\begin{figure}[ht]
    \centering
    \includegraphics[width=0.85\linewidth]{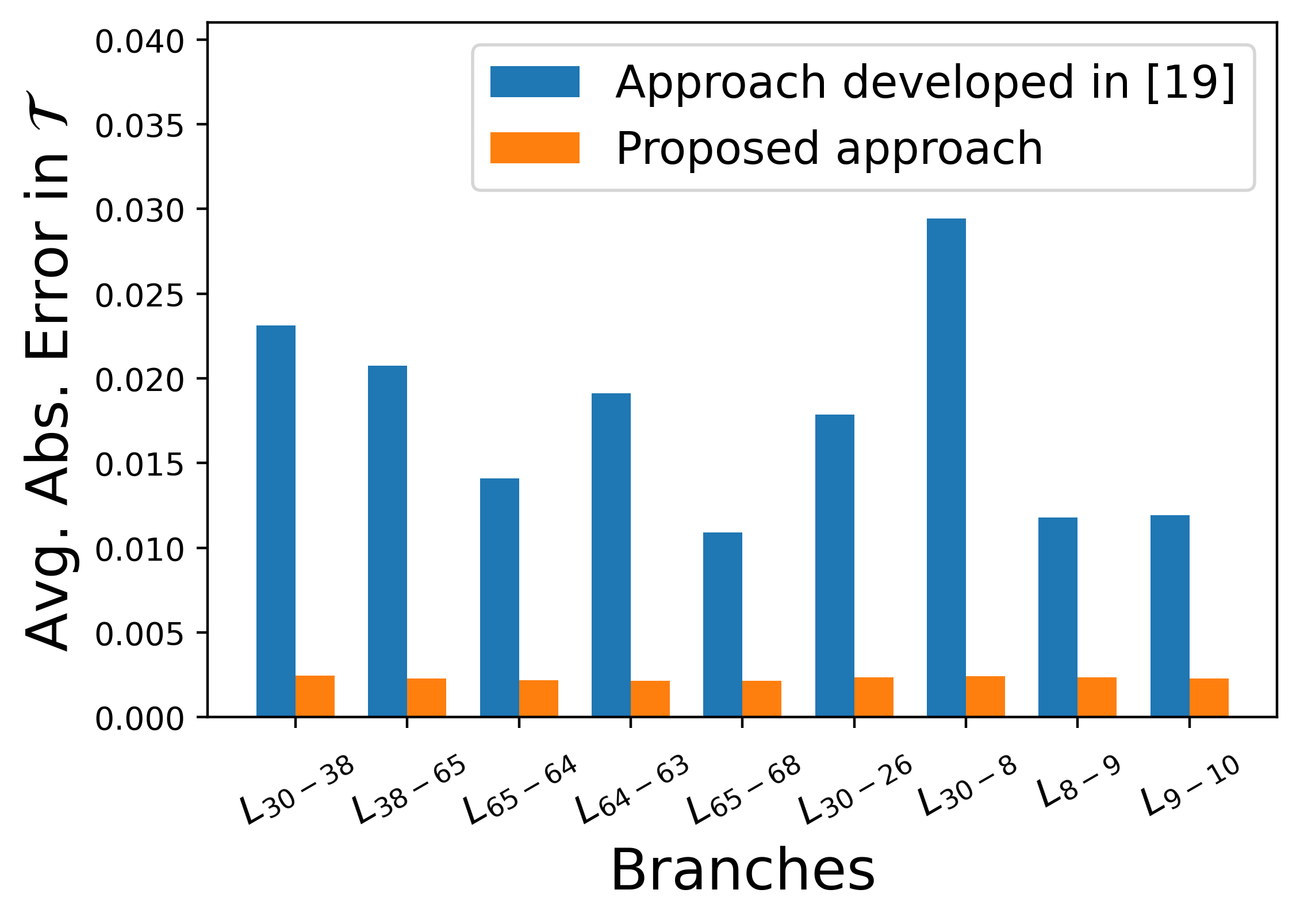}
    \vspace{-0.5em}
    \caption{Comparing IT correction factor estimates with \cite{wang2019transmission} }
\label{Impact_IT_class_CFE_comparison_chenwang}
\end{figure}

Fig. \ref{Impact_IT_class_LPE_comparison_chenwang} displays the average of the $\mathrm{MARE}$s of the three line parameters for every line. 
A visual inspection of the bar chart 
confirms that the proposed approach comprehensively outperforms the approach developed in \cite{wang2019transmission} for any branch that is more than one hop away from the RQM branch. 
Quantitatively, the average LPE error in \cite{wang2019transmission} was 2.07\%, whereas it was only 0.34\% for the proposed approach.
Similarly, Fig. \ref{Impact_IT_class_CFE_comparison_chenwang} compares the average of the absolute error of IT correction factor estimates for every branch except the $L_{68-81}$ branch (which has the reference RQM).
It can be noticed that for every branch, the results obtained using the proposed method are much better in comparison
to that obtained from \cite{wang2019transmission}. 
This can be further validated from the fact that the average of the absolute error across all the branches 
was $0.017$ in the case of \cite{wang2019transmission}, while it was only $0.0022$ for the proposed approach.

\subsection{Performing SLIC with Field PMU Data}
In this study, we applied the proposed approach to actual PMU data acquired from a power utility in the Eastern Interconnection of the U.S.
The test was performed on a 
part of the network comprising four buses (A, B, C, and D) joined by three branches (A-B, B-C, and C-D). 
One
obstacle that was encountered when working with
real-world data was the lack of knowledge of the ground truth.
To 
circumvent this problem,
we adopted a validation methodology that focused on the \textit{consistency} of the estimates. 
The methodology involved segregating the collected measurement data into two sets. The first set, designated as $S_1$, comprises voltage and current measurements recorded on Mondays, Wednesdays, and Fridays over a period of two consecutive weeks. In contrast, the second set, $S_2$, consists of similar measurements (in terms of time of day) but captured on Tuesdays and Thursdays of three consecutive weeks. 
A comparison of the results that were obtained when the proposed approach and validation methodology were applied to the two datasets for LPE and IT correction factor estimation are presented in Table \ref{OGnE_SLIC_LPE} and Table \ref{OGnE_SLIC_CFE}, respectively.

It can be observed
from Table \ref{OGnE_SLIC_LPE} that the line parameter estimates from $S_1$ and $S_2$ are identical, which validates the estimation consistency.
From Table \ref{OGnE_SLIC_CFE}, it can be observed that
the estimates derived from $S_1$ and $S_2$ for the VTs exhibit a very high degree of similarity (see entries in columns 3 and 4).
The corresponding entries for the CTs show a similar pattern, but with minor variations (see columns 5 and 6); this is due to the reason mentioned in Section  \ref{SLIC - Performance Analysis - HV Network results} (CT correction factors require knowledge of additional quantities).
Collectively, Tables \ref{OGnE_SLIC_LPE} and \ref{OGnE_SLIC_CFE}: (a) support the notion that correction factors only change periodically, implying that solving the SLIC problem every few months is sufficient, and (b) substantiate the effectiveness of the proposed solution to the SLIC problem for actual power system conditions.

\begin{table}
\centering
\caption{Line parameters estimates obtained using proposed approach with real PMU data}
\label{OGnE_SLIC_LPE}
\scriptsize
\centering
\begin{tabular}{|c|c|c|c|c|}
\hline
\multicolumn{1}{|l|}{Branch} &  Set  & r (in p.u.)       & x (in p.u.)     & b (in p.u.)      \\ \hline
\multirow{2}{*}{A-B}  & $S_1$ & 0.00238 & 0.0315 & 0.3503 \\ \cline{2-5} 
                       & $S_2$ & 0.00238 & 0.0315 & 0.3503 \\ \hline
\multirow{2}{*}{B-C}  & $S_1$ & 0.00384 & 0.0518 & 0.5755 \\ \cline{2-5} 
                       & $S_2$ & 0.00384 & 0.0518 & 0.5755 \\ \hline
\multirow{2}{*}{C-D}   & $S_1$ & 0.00269 & 0.0248 & 0.43   \\ \cline{2-5} 
                       & $S_2$ & 0.00269 & 0.0248 & 0.43   \\ \hline
\end{tabular}
\end{table}

\begin{table}
\centering
\caption{IT correction factor estimates obtained using proposed approach with real PMU data}
\label{OGnE_SLIC_CFE}
\scriptsize
\centering
\begin{tabular}{|c|c|c|c|c|c|}
\hline
\multicolumn{1}{|l|}{Branch} & Set   & $\hat{\tau}_{V_{pq}}$           & $\hat{\tau}_{V_{qp}}$            & $\hat{\tau}_{I_{pq}}$              & $\hat{\tau}_{I_{qp}}$          \\ \hline
\multirow{2}{*}{A-B}  & $S_1$ & -- --     & 1.00-0.01j & 1.13+0.01j     & 1.13+0.01j \\ \cline{2-6} 
                       & $S_2$ & -- --        & 1.01-0.00j   & 1.12+0.01j     & 1.13+0.00j \\ \hline
\multirow{2}{*}{B-C}  & $S_1$ & 1.01+0.00j & 1.01+0.01j  & 1.09+0.01j & 1.09+0.02j \\ \cline{2-6} 
                       & $S_2$ & 1.01+0.01j & 1.01+0.01j   & 1.10+0.02j    & 1.09+0.02j \\ \hline
\multirow{2}{*}{C-D}   & $S_1$ & 1.00+0.01j & 0.99+0.01j   & 1.12-0.02j    & 1.11-0.01j  \\ \cline{2-6} 
                       & $S_2$ & 1.01+0.01j & 0.99+0.02j  & 1.13-0.03j    & 1.13-0.00j \\ \hline
\end{tabular}
\end{table}

\section{Conclusion}
\label{Conclusion}
This paper presents a method to simultaneously estimate line parameters and calibrate ITs in a power system that has a connected tree. 
The proposed approach employs TLS, which is a statistical method, in combination with a novel quantization procedure to uniquely estimate line parameters and correction factors of a single branch. The methodology is then extended to perform system-wide SLIC
by exploiting the multiple independent observations
of the bus voltages by different PMUs.
The results obtained using the IEEE 118-bus system confirmed exceptional estimation accuracy under normal conditions.
Sensitivity studies conducted by
varying the 
additive noise and considering
different IT accuracy classes indicated stable performance under diverse conditions.
A strategy
to find the optimal RQM location 
for solving the SLIC problem was also successfully demonstrated.
Finally, consistency of the estimates while using the proposed solution with field PMU data validated its usability in practical scenarios. 

\appendix

\setcounter{equation}{0}
\numberwithin{equation}{subsection}

\subsection{Proof of Theorem \ref{Theorem_1}}\label{appendix1}

\begin{proof}
Let the resistance, reactance, and shunt susceptance of a transmission line under consideration be $r$, $x$, and $b$, respectively.
Assume that for two bin numbers $m_1$ and $m_2$, $f_W(m_1) = f_W(m_2)$. Then, by definition (see \eqref{f_W_definition}), we have:
\begin{equation}
\label{Unique_eta_proof_a}
\begin{aligned}
& f_W(m_1) = f_W(m_2) \\
\implies   &
1 - x(m_1)b(m_1) + j (r(m_1)b(m_1)) =  1 - x(m_2)b(m_2)  \\
&+ j (r(m_2)b(m_2)).
\end{aligned}
\end{equation}

Separating the real and imaginary parts of \eqref{Unique_eta_proof_a}, and substituting the values of $r(m)$, $x(m)$, and $b(m)$, we get, 
\begin{equation}
\label{Unique_eta_proof_b}
\begin{aligned}
 x_{\dagger} (1+ \delta_{x} m_1)  b_{\dagger} (1+ \delta_{b} m_1) &=  x_{\dagger} (1+ \delta_{x} m_2)  b_{\dagger} (1+ \delta_{b} m_2) \\
   r_{\dagger} (1+ \delta_{r} m_1)  b_{\dagger} (1+ \delta_{b} m_1) &=  r_{\dagger} (1+ \delta_{r} m_2)  b_{\dagger} (1+ \delta_{b} m_2).
\end{aligned}
\end{equation}

Rearranging terms of the first sub-equation of \eqref{Unique_eta_proof_b}, and considering the fact that $x_{\dagger} \neq 0$, $r_{\dagger} \neq 0$, and $b_{\dagger} \neq 0$, we get,
\begin{equation}
\label{Unique_eta_proof_c}
\begin{aligned}
&(\delta_{x} m_1 + \delta_{b} m_1 + \delta_{x} \delta_{b} m_1^2) = (\delta_{x} m_2 + \delta_{b} m_2 + \delta_{x} \delta_{b} m_2^2) \\
&\delta_{x} (m_1 - m_2)+ \delta_{b} (m_1 - m_2) + \delta_{x} \delta_{b}  (m_1^2 - m_2^2)  =0 \\
& (m_1 - m_2) \left[ \delta_{x}+ \delta_{b}  + \delta_{x} \delta_{b}  (m_1 + m_2) \right] =0.
\end{aligned}
\end{equation}

Similarly, rearranging the terms of the second sub-equation of \eqref{Unique_eta_proof_b}, we get,
\begin{equation}
\label{Unique_eta_proof_d}
\begin{aligned}
&(\delta_{r} m_1 + \delta_{b} m_1 + \delta_{r} \delta_{b} m_1^2) = (\delta_{r} m_2 + \delta_{b} m_2 + \delta_{r} \delta_{b} m_2^2) \\
&\delta_{r} (m_1 - m_2)+ \delta_{b} (m_1 - m_2) + \delta_{r} \delta_{b}  (m_1^2 - m_2^2)  =0 \\
& (m_1 - m_2) \left[ \delta_{r}+ \delta_{b}  + \delta_{r} \delta_{b}  (m_1 + m_2) \right] =0. \hspace{-1.5cm}
\end{aligned}
\end{equation}

Both \eqref{Unique_eta_proof_c} and \eqref{Unique_eta_proof_d} must be simultaneously satisfied for $f_W(m_1)=f_W(m_2)$. Consequently, two scenarios emerge. 
The first scenario is  $m_1 = m_2$. The second scenario is $\left[ \delta_{x}+ \delta_{b}  + \delta_{x}  \delta_{b}  (m_1 + m_2)  \right] =0$, and $ \left[ \delta_{r}+ \delta_{b}  + \delta_{r} \delta_{b}  (m_1 + m_2) \right] =0 $. Let  \mbox{$u_1 =  \delta_{b}m_1$,}  and \mbox{$u_2 =  \delta_{b} m_2$}. Then, the second scenario becomes $\left[ \delta_{x}+ \delta_{b}  + \delta_{x}   (u_1 + u_2)  \right] =0$, and $ \left[ \delta_{r}+ \delta_{b}  + \delta_{r}   (u_1 + u_2) \right] =0 $.
This scenario can now be reduced to
\begin{equation}
\label{Unique_eta_proof_e}
\begin{aligned}
\delta_{r} = \delta_{x} \text{ AND } u_1 + u_2 = \frac{-\delta_r - \delta_b}{\delta_r}.
\end{aligned}
\end{equation}

Note that $u_1$ and $u_2$ can take real values in the interval $[-0.3, 0.3]$ (see second paragraph of Section \ref{QP4IBSLIC}).
Now, since there are infinite real numbers between $[-0.3, 0.3]$, the probability of the sum of $u_1$ and $u_2$ being equal to a specific number (in this case, $(-\delta_r - \delta_b)/\delta_r$), is zero. That is, 
\begin{equation}
\label{Unique_eta_proof_f}
\begin{aligned}
\mathbb{P} \left[ u_1 + u_2 = \frac{-\delta_r - \delta_b}{\delta_r} \right] =
0\\
\end{aligned}
\end{equation}

Now, an event that happens with probability zero happens \textit{almost never} \cite{gradel2007finite}.
Therefore, from \eqref{Unique_eta_proof_f} it can be inferred that the second scenario (given by \eqref{Unique_eta_proof_e}) is unlikely to ever occur.
Hence, the first scenario,  \mbox{$ (m_1 - m_2) = 0$}, is the only practical possibility.
This proves that 
\begin{equation}
\begin{aligned}
\hspace{-5pt} &f_W(m_1) = f_W(m_2) \implies m_1 = m_2 \implies \eta_{m_1} = \eta_{m_2} \hspace{-0.5cm}
 \end{aligned}
\end{equation}

The converse is trivial:
\begin{equation}
\resizebox{0.95\hsize}{!}{$
\begin{aligned}
 & m_1 = m_2  \implies f_W(m_1) = f_W(m_2)  \:\: \because \text{  $f_W$ definition} \\
 \therefore & f_W(m_1) = f_W(m_2) \iff m_1 = m_2 \iff \eta_{m_1} = \eta_{m_2}
 \end{aligned}
  $}
\end{equation}

To summarize, the $W$ values of two bin numbers are the same only if the bin numbers are the same, and correspondingly have the same line parameter values. This, in turn, implies that the proposed quantization procedure guarantees the unique determination of the line parameters.  
\end{proof}

\bibliographystyle{IEEEtran}
\bibliography{ References/SLIC_Meta_References}

\end{document}